\documentclass[11pt,a4paper]{article}
\usepackage{tikz}
\usepackage{amsmath,amssymb,amsfonts,amsthm}

\numberwithin{equation}{section}

\newcommand{\be}{\begin{equation}}
\newcommand{\ee}{\end{equation}}

\newcommand{\bea}{\begin{eqnarray}}
\newcommand{\eea}{\end{eqnarray}}
\newcommand{\bse}{\begin{subequations}}
\newcommand{\ese}{\end{subequations}}

\newtheorem{theorem}{Theorem}[section]

\newtheorem{proposition}[theorem]{Proposition}
\newtheorem{lemma}[theorem]{Lemma}

\newtheorem{remark}[theorem]{Remark}

\newcommand{\C}{{\mathbb C}}

\makeatletter
 
 \newcommand{\Rmnum}[1]{\expandafter\@slowromancap\romannumeral #1@}
 \makeatother

\newcommand{\ba}{\begin{array}}
\newcommand{\ea}{\end{array}}
\def\bee{\begin{eqnarray}}
\def\ene{\end{eqnarray}}
\def\bes{\begin{subequations}}
\def\ees{\end{subequations}}

\newcommand{\ol}{\overline}

\newcommand{\id}{\mathbb{I}}

\newcommand{\re}{\mathrm{Re}}

\newcommand{\eps}{\varepsilon}

\newcommand{\Sig}{\Sigma}

\newcommand{\Lam}{\Lambda}
\newcommand{\gam}{\gamma}

\newcommand{\Om}{\Omega}

\newcommand{\dta}{\delta}
\newcommand{\Dta}{\Delta}

\newcommand{\tha}{\theta}

\numberwithin{equation}{section}

\title{Initial-boundary value problem for the two-component Gerdjikov-Ivanov equation on the interval}

\author{
Qiaozhen Zhu$^{a}$ \footnote{E-mail address: qiaozhenzhu13@fudan.edu.cn}, ~~Jian Xu $^{b}$ \footnote{ jianxu@usst.edu.cn}, ~~En-gui Fan$^{a}$ \footnote{correspondence author. E-mail address: faneg@fudan.edu.cn}\\
{\small \it $^{a}$School of Mathematical Sciences, Fudan University, Shanghai 200433, P. R. China} \\
{\small \it $^{b}$College of Science, University of Shanghai for Science and Technology, Shanghai 200093, P. R. China} \\
}
\date{
}
\begin{document}

\maketitle

\begin{abstract}
\baselineskip=18pt

 In this paper, we apply Fokas unified  method  to study initial-boundary value problems for the two-component Gerdjikov-Ivanov equation formulated on the finite interval with $3 \times 3$ Lax pairs. The solution can be expressed in terms of the solution of a $3\times3$ Riemann-Hilbert problem. The relevant jump matrices are explicitly given in terms of three matrix-value spectral functions $s(\lambda)$, $S(\lambda)$ and $S_L(\lambda)$, which arising from the initial values at $t=0$, boundary values at $x=0$ and boundary values at $x=L$, respectively. Moreover, The associated  Dirichlet to Neumann map is  analyzed via the global relation. The relevant formulae for boundary value problems on the finite interval can reduce to ones on the half-line as the length of the interval tends to infinity.

\vskip 6pt
\noindent{\textbf{Keywords:}}  Two-component Gerdjikov-Ivanov equation, initial-boundary value problem, Fokas unified method, Riemann-Hilbert problem.\\
\end{abstract}
\vskip 20pt

%==================================introduction===============================================
\baselineskip=18pt

\section{Introduction}

The  Gerdjikov-Ivanov (GI) equation takes in  the form   \cite{gi}
\begin{equation}\label{GI}
q_t=iq_{xx}+q^2\bar q_x+\frac{i}{2}q^3{\bar q}^2.
\end{equation}
In these  years, there has been much work on   the GI equation, including   Hamiltonian structures \cite{fan2}, Darboux transformation \cite{fan1}, rouge wave and breather soliton \cite{he1}, algebro-geometric solutions \cite{hfan}, envelope bright and dark soliton solution \cite{mawx}. Recently, Zhang,Cheng and He obtained the N-soliton solutions with Riemann-Hilbert method about the two-component (2-GI) equation \cite{hjs}
\be\label{cGI}
\left\{
\ba{l}
q_{1t}=iq_{1xx}+q_1(q_1\bar{q}_{1x}+q_2\bar{q}_{2x})+\frac{i}{2}q_1(|q_1|^4+|q_2|^4)+i|q_1q_2|^2q_1,\\
q_{2t}=iq_{2xx}+q_2(q_1\bar{q}_{1x}+q_2\bar{q}_{2x})+\frac{i}{2}q_2(|q_1|^4+|q_2|^4)+i|q_1q_2|^2q_2.
\ea
\right.
\ee

\par
%A general method for solving boundary value problems for linear and integral nonlinear PDEs was announced by Fokas in \cite{f1} and further developed in \cite{f3, f4}
 In 1997, Fokas announced the unified transform for the analysis of initial boundary value (IBV) problems for linear and nonlinear integrable PDEs \cite{f1}.  The Fokas method  was usually used to analyze the IBV problem for integrable PDEs with $2\times2$ Lax pair on the half-line and  the finite interval, such as  nonlinear Schr\"oding equation \cite{nls1, nls2},  sine-Gordon equation \cite{sg1, sg2},  KdV equation \cite{kdv},  mKdV equation \cite{mkdv1, mkdv2},  derivative nonlinear Schr\"oding equation \cite{dnls}. In 2012, Lenells extended this method to the IBV problem of  integrable systems  with $3\times3$ Lax pair on the half-line \cite{l1}.  After that, several important integrable equations with $3\times3$ Lax pair have been investigated, including
Degasperis-Procesi \cite{dp}, Sasa-satuma \cite{ss}. However,  there has been still  less work  on the IBV problems on the finite interval  of  integrable equations with $3 \times 3$ Lax pair except  to  the two-component NLS  \cite{xf2nlsf}, general coupled NLS \cite{tsf} and the integrable spin-1 Gross-Pitaevskii \cite{yzy} equations.
%In \cite{zfx}, the IBV problems were considered  on the half-line. It is a natural problem to extend the Fokas methodology to the case of initial boundary value problems on the interval.
%There are two differences for solving the IBV problems of  nonlinear integrable equations between on the interval and on  the half-line. Firstly, we should make a distinction between the integration contour $\gamma_3$ and $\gamma_4$ when we try to analyze the IBV problem on the interval, while in the half-line case, there just one integration curve $\gamma_3$. Secondly, In the interval case, we need to introduce a new factor $\frac{1}{\Delta}$ in (\ref{DelSig}) during analyzing the global relation to characterize the unknow boundary data in terms of the given initial and boundary vaule data.

In this paper, we apply Fokas method to consider 2-GI equation with the following   initial boundary value data:
%We will denote the initial data of (\ref{cGI}) by $\{q_{10}(x), q_{20}(x)\}$, while the Dirichlet and Neumann boudary values will be denoted by
%$\{g_{01}(t), g_{02}(t), f_{01}(t), f_{02}(t)\}$ and $\{g_{11}(t), g_{12}(t), f_{11}(t), f_{12}(t)\}$, respectively, i.e.
\begin{equation}\label{ibv-cgi}
\ba{llll}
&\mbox{Initial value:}  &q_1(x,t=0)=q_{10}(x),\quad &q_2(x,t=0)=q_{20}(x),\\
&\mbox{Dirichlet boundary value:}  &q_1(x=0,t)=g_{01}(t),\quad &q_2(x=0,t)=g_{02}(t),\\
   &&q_1(x=L,t)=f_{01}(t),\quad &q_2(x=L,t)=f_{02}(t),\\
&\mbox{Neumann boundary value:} &q_{1x}(x=0,t)=g_{11}(t),\quad  &q_{2x}(x=0,t)=g_{12}(t),\\
  &&q_{1x}(x=L,t)=f{11}(t),\quad  &q_{2x}(x=L,t)=f_{12}(t),
\ea
\end{equation}
where $q_1(x,t)$ and $q_2(x,t)$ are complex-valued functions of $(x,t) \in \Omega$,  and   $\Omega$ denotes  the finite interval domain
\[
\Omega=\{(x,t)| 0\leq x \leq L, 0\leq t \leq T\},
\]
here $L>0$ is a positive fixed constant and $T>0$ being a fixed final time.

\par
Comparing with two-component NLS  equation \cite{xf2nlsf},   the IBV problem of the 2-GI equation  (\ref{cGI})  also presents some distinctive features in the use  of Fokas method:
 (i) The order of spectral variable $k$ in the Lax pair (\ref{Lax}) is higher than that  of   2-NLS  equation. In order to make the results on the interval reduce to the ones on the half-line,   we should first introduce transformation  $\psi(x,t, k)=k^{\frac{1}{2}\Lambda}\phi(x,t,k) k^{-\frac{1}{2}\Lambda}$ so that the Lax pairs are even functions of $k$.
 (ii)   The  2-GI equation admits
a generalized Wadati-Konno-Ichikawa (WKI) type Lax pair,  %since the spectral variable $\lambda$ is the same order in the  Lax pair (\ref{Lax1'}). %while the 2-NLS equation admits a AKNS-type Lax pair.
     which  admits  a gauge  transformation to  AKNS-type Lax pair,  but this gauge transformation can not be used to analyze the IBV problem by mapping it into  2-NLS equation. We need to introduce a matrix-value function $G(x,t)$ to transform the WKI-type Lax pair into AKNS-type Lax pair.
%Define a $3\times 3$ matrix-value function G(x,t), let $ \phi(x,t,k)=G(x,t)\mu(x,t,k) e^{-i\Lam \lambda x-2i\Lam \lambda^2 t}$, we get an equivalent Lax pair (\ref{muLax}). This is the desired Lax pair.

\par
 Organization of  this  paper is as follows.   In the following  section 2,  we perform the spectral analysis of the associated Lax pair
 for the 2-GI equation (\ref{cGI}).   In the section 3, we  give  the
corresponding  matrix RH problem  associated with  the IBV problem of  2-GI equation.
In section 4,  we get the map between the Dirichlet and the Neumann boundary problem through analysising the global relation.
Especially, the
relevant formulae for boundary value problems on the finite interval can reduce to ones on the half-line as the
length of the interval approaches to infinity.

\section{Spectral analysis}

\subsection{Lax pair}
The 2-GI equation admits a  $3 \times 3$ Lax pair \cite{hjs}
\begin{subequations}\label{Lax}
\be\label{Lax-x}
\psi_x+ik^2\Lam\psi=U_1\psi,
\ee
\be\label{Lax-t}
\psi_t+2ik^4\Lam\psi=U_2\psi,
\ee
\end{subequations}
where $\psi(x,t,k)$ is a $3\times 3-$matrix valued  eigenfunction, $k \in \mathbb{C}$ is the spectral parameter, and $U_1(x,t), U_2(x,t)$ are $3\times 3-$matrix valued functions given by
\begin{equation}\label{U1U2def}
U_1=-kQ\Lam+\frac{i}{2}Q^2\Lam, \qquad U_2=-2k^3Q\Lam+ik^2\Lam Q^2+ikQ_x-\frac{1}{2}[Q_x,Q]+\frac{i}{4}Q^4\Lam,
\end{equation}
\be\label{Lamdef}
\Lam=\left(\ba{ccc}1&0&0\\0&-1&0\\0&0&-1\ea\right), \qquad \qquad Q=\left(\ba{ccc}0&q_1&q_2\\\bar q_1&0&0\\\bar q_2&0&0\ea\right).
\ee

%\subsection{The first transformation}
\par
There are both odd power and even power of $k$ in the Lax pair (\ref{Lax}),  to make (\ref{Lax}) are even functions of $k$ for analyzing  the large $L$ limit, we   introduce a transformation
\begin{equation}
  \psi(x,t, k)=k^{\frac{1}{2}\Lambda}\phi(x,t,k) k^{-\frac{1}{2}\Lambda},
\end{equation}
and  get an equivalent Lax pair
\begin{subequations}\label{Lax1}
\be\label{Lax1-x}
\phi_x+ik^2\Lam\phi=\tilde{U}_1\phi,
\ee
\be\label{Lax1-t}
\phi_t+2ik^4\Lam\phi=\tilde{U}_2\phi,
\ee
\end{subequations}
where
\begin{equation}\label{tildeU1U2def}
\tilde{U}_1=-Q_1\Lam-k^2Q_2\Lam+\frac{i}{2}Q^2\Lam,
 \qquad \tilde{U}_2=-2k^4Q_2\Lambda+k^2(i\Lam Q^2+iQ_{2x}-2Q_1\Lambda)+(iQ_{1x}-\frac{1}{2}[Q_x,Q]+\frac{i}{4}Q^4\Lam),
\end{equation}
\be\label{Q1Q2def}
Q_1=\left(\ba{ccc}0&q_1&q_2\\0&0&0\\0&0&0\ea\right), \qquad \qquad Q_2=\left(\ba{ccc}0&0&0\\\bar q_1&0&0\\\bar q_2&0&0\ea\right), \qquad Q=Q_1+Q_2.
\ee
Let $\lambda=k^2$, Lax pair (\ref{Lax1}) becomes
\begin{subequations}\label{Lax1'}
\be\label{Lax1'-x}
\phi_x+i\lambda\Lam\phi=\tilde{U}_1\phi,
\ee
\be\label{Lax1'-t}
\phi_t+2i\lambda^2\Lam\phi=\tilde{U}_2\phi,
\ee
\end{subequations}
where $\tilde{U}_1, \tilde{U}_2$ are given by (\ref{tildeU1U2def}) with $k^2$ replaced with $\lambda$.
\subsection{The closed one-form}

\par

Defining  a $3\times3$ matrix-value function %\cite{xfsp}
\begin{equation}\label{Gdef}
G(x,t)=\left(\ba{ccc}1&0&0\\\frac{1}{2i}\bar q_1&1&0\\\frac{1}{2i}\bar q_2&0&1\ea\right),
\end{equation}
and making    a transformation
\begin{equation}
  \phi(x,t,k)=G(x,t)\mu(x,t,k) e^{-i\lambda\Lam x-2i\lambda^2\Lam t},
\end{equation}
then we get a new Lax pair for $\mu(x,t,\lambda)$
\begin{subequations}\label{muLax}
\begin{equation}
\mu_x+i\lambda[\Lam,\mu]=V_1\mu,
\ee
\be
\mu_t+2i\lambda^2[\Lam,\mu]=V_2\mu,
\ee
\end{subequations}
where
\begin{subequations}\label{V1V2def}
\begin{equation}\label{V1def}
V_1=G^{-1}(-Q_1\Lam+\frac{i}{2}Q^2\Lam)G-G^{-1}G_x,
\end{equation}
\begin{equation}\label{V2def}
V_2=\lambda G^{-1}(i\Lam Q^2+iQ_{2x}-2Q_1\Lambda)G+G^{-1}(iQ_{1x}-\frac{1}{2}[Q_x,Q]+\frac{i}{4}Q^4\Lam)G-G^{-1}G_t.
\end{equation}
\end{subequations}
%\par
%Let $\lambda=k^2$, then we can write (\ref{muLax})
%\begin{subequations}\label{muLaxe}
%\be\label{muLax-x}
%\mu_x+i\lambda[\Lam,\mu]=V_1\mu,
%\ee
%\be\label{muLax-t}
%\mu_t+2i\lambda^2[\Lam,\mu]=V_2\mu,
%\ee
%\end{subequations}

\par

Letting $\hat A$ denotes the operators which acts on a $3\times 3$ matrix $X$ by $\hat A X=[A,X]$ , then   the equations   (\ref{muLax}) can be
rewritten in a differential form
\be\label{mudiffform}
d(e^{(i\lambda x+2i\lambda^2t)\hat \Lam}\mu)=W,
\ee
where the closed one-form $W(x,t,k)$ is defined by
\be\label{Wdef}
W=e^{(i\lambda x+2i\lambda^2t)\hat \Lam}(V_1dx+V_2dt)\mu.
\ee
%Equations (\ref{muLaxe}) are the coordinate form of equation (\ref{mudiffform}); However, for the implementation of Fokas method, it is often convenient to use equation (\ref{mudiffform}).

\subsection{The eigenfunctions $\mu_j$'s}
We define four eigenfunctions $\{\mu_j\}_1^4$ of (\ref{muLax}) by the Volterra integral equations
\be\label{mujdef}
\mu_j(x,t,k)=\id+\int_{\gam_j}e^{-(i \lambda x+2i \lambda^2 t)\hat \Lam}W_j(x',t',k).\qquad j=1,2,3,4.
\ee
where $W_j$ is given by (\ref{Wdef}) with $\mu$ replaced by $\mu_j$,  and the contours $\{\gam_j\}_1^4$
 can be given by the following inequalities ( see  Figure 1):
\be
\ba{ll}
\gam_1:&x-x'\ge 0,t-t'\le 0,\\
\gam_2:&x-x'\ge 0,t-t'\ge 0,\\
\gam_3:&x-x'\le 0,t-t'\ge 0,\\
\gam_4:&x-x'\le 0,t-t'\le 0.\\
\ea
\ee

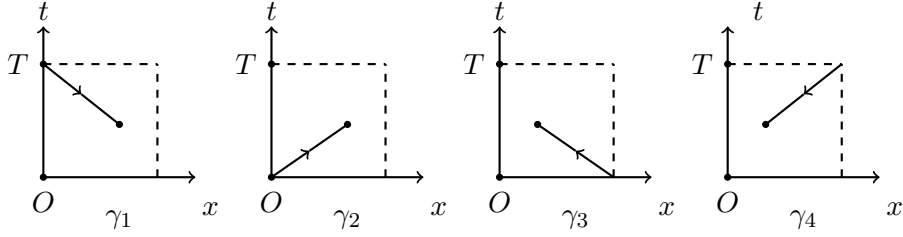
\begin{figure}
\begin{center}
\begin{tikzpicture}
\draw[thick,->](-6,0)--(-4,0);
\draw[thick,->](-6,0)--(-6,2);
\draw[thick,dashed](-6,1.5)--(-4.5,1.5);
\draw[thick,dashed](-4.5,0)--(-4.5,1.5);
\draw[thick,->](-6,1.5)--(-5.5,1.1);
\draw[thick,-](-5.5,1.1)--(-5,.7);
\draw [fill] (-6,0) circle [radius=0.04];
\draw [fill] (-6,1.5) circle [radius=0.04];
\draw [fill] (-5,.7) circle [radius=0.04];
\draw (-6,-0.06) node[below]{$O$};
\draw (-6.06,1.5) node[left] {$T$};
\draw (-5,-0.3) node[below] {$\gamma_1$} ;
\draw(-3.8,-.2) node[below] {$x$};
\draw(-6.2,2.2) node[right] {$t$};

\draw[thick,->](-3,0)--(-1,0);
\draw[thick,->](-3,0)--(-3,2);
\draw[thick,dashed](-3,1.5)--(-1.5,1.5);
\draw[thick,dashed](-1.5,0)--(-1.5,1.5);
\draw[thick,->](-3,0)--(-2.5,0.35);
\draw[thick,-](-2.5,0.35)--(-2,.7);
\draw [fill] (-3,0) circle [radius=0.04];
\draw [fill] (-3,1.5) circle [radius=0.04];
\draw [fill] (-2,.7) circle [radius=0.04];
\draw (-3,-0.06) node[below]{$O$};
\draw (-3.06,1.5) node[left] {$T$};
\draw (-2,-0.3) node[below] {$\gamma_2$} ;
\draw(-.8,-.2) node[below] {$x$};
\draw(-3.2,2.2) node[right] {$t$};

\draw[thick,->](0,0)--(2,0);
\draw[thick,->](0,0)--(0,2);
\draw[thick,dashed](0,1.5)--(1.5,1.5);
\draw[thick,dashed](1.5,0)--(1.5,1.5);
\draw[thick,->](1.5,0)--(1,.35);
\draw[thick,-](1,.35)--(.5,.7);
\draw [fill] (0,0) circle [radius=0.04];
\draw [fill] (0,1.5) circle [radius=0.04];
\draw [fill] (.5,.7) circle [radius=0.04];
\draw (0,-0.06) node[below]{$O$};
\draw (-0.06,1.5) node[left] {$T$};
\draw (1,-0.3) node[below] {$\gamma_3$} ;
\draw(2.2,-.2) node[below] {$x$};
\draw(-.2,2.2) node[right] {$t$};

\draw[thick,->](3,0)--(5,0);
\draw[thick,->](3,0)--(3,2);
\draw[thick,dashed](3,1.5)--(4.5,1.5);
\draw[thick,dashed](4.5,0)--(4.5,1.5);
\draw[thick,->](4.5,1.5)--(4,1.1);
\draw[thick,-](4,1.1)--(3.5,.7);
\draw [fill] (3,0) circle [radius=0.04];
\draw [fill] (3,1.5) circle [radius=0.04];
\draw [fill] (3.5,.7) circle [radius=0.04];
\draw (3,-0.06) node[below]{$O$};
\draw (2.94,1.5) node[left] {$T$};
\draw (4,-0.3) node[below] {$\gamma_4$} ;
\draw(5.2,-.2) node[below] {$x$};
\draw(3.2,2.2) node[right] {$t$};
\end{tikzpicture}
\end{center}
\caption{The four contours $\gamma_1,\gamma_2, \gamma_3$ and $\gamma_4$ in the $(x,t)-$domain.}
\label{FIG1}
\end{figure}
and the matrix equation (\ref{mujdef}) involves the exponentials
\be
\ba{ll}
\mbox{$[\mu_j]_1$:}&e^{2i\lambda(x-x')+4i\lambda^2(t-t')},e^{2i\lambda(x-x')+4i\lambda^2(t-t')}\\
\mbox{$[\mu_j]_2$:}&e^{-2i\lambda(x-x')-4i\lambda^2(t-t')},\\
\mbox{$[\mu_j]_3$:}&e^{-2i\lambda(x-x')-4i\lambda^2(t-t')}.
\ea
\ee
from which,  we find that   the functions $\{\mu_j\}_1^4$ are bounded and analytic for $\lambda \in\C$ such that $\lambda$ belongs to
\be\label{mujbodanydom}
\ba{ll}
\mu_1:&(D_2,D_3,D_3),\\
\mu_2:&(D_1,D_4,D_4),\\
\mu_3:&(D_3,D_2,D_2),\\
\mu_4:&(D_4,D_1,D_1),\\
\ea
\ee
where $\{D_n\}_1^4$ denote four open, pairwisely disjoint subsets of the complex $\lambda-$plane showed in Figure 2.
%\begin{figure}
%\begin{center}
%\begin{tikzpicture}
%%\draw[thick,-] (-4,0)--(0,0)node[below] {$k_0$}--(4,0);
%%\draw [thick,-] (-2,3) [out=-60, in=60] to (-2,-3);
%%\draw [thick,-] (2,3)[out=-120, in=120] to (2,-3);
%\draw[thick,->](0,-3)--(0,3);
%\draw[thick,->](-3,0)--(3,0);
%\draw[thick,-](-2.5,-2.5)--(2.5,2.5);
%\draw[thick,-](-2.5,2.5)--(2.5,-2.5);
%\draw (.3,-.2) node[below] {0};
%\draw (1.8,.75) node[left] {$D_1$};
%\draw (1,1.5) node[left] {$D_2$};
%\draw (-.6,1.5) node[left] {$D_3$};
%\draw (-1.6,.75) node[left] {$D_4$};
%\draw (-1.6,-.75) node[left] {$D_1$};
%\draw (-.6,-1.5) node[left] {$D_2$};
%\draw (1,-1.5) node[left] {$D_3$};
%\draw (1.8,-.75) node[left] {$D_4$};
%\draw(3.5,0) node {$Re k$};
%\draw(0,3.3) node {$Im k$};
%\end{tikzpicture}
%\end{center}
%\caption{The sets $D_n, n=1,\ldots,4$, which decompose the complex k-plane.}
%\label{FIGURE 2}
%\end{figure}
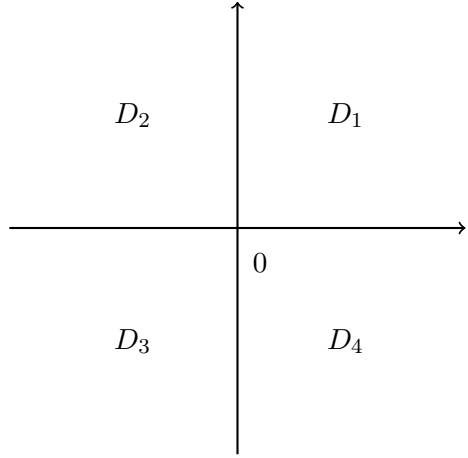
\begin{figure}
\begin{center}
\begin{tikzpicture}
%\draw[thick,-] (-4,0)--(0,0)node[below] {$k_0$}--(4,0);
%\draw [thick,-] (-2,3) [out=-60, in=60] to (-2,-3);
%\draw [thick,-] (2,3)[out=-120, in=120] to (2,-3);
\draw[thick,->](0,-3)--(0,3);
\draw[thick,->](-3,0)--(3,0);
\draw (.3,-.2) node[below] {0};
\draw (1.8,1.5) node[left] {$D_1$};
%\draw (4,1) node[left] {$\re iF < 0$};
\draw (-1,1.5) node[left] {$D_2$};
\draw (-1,-1.5) node[left] {$D_3$};
\draw (1.8,-1.5) node[left] {$D_4$};
\end{tikzpicture}
\end{center}
\caption{The sets $D_n, n=1,\ldots,4$, which decompose the complex $\lambda$-plane.}
\label{FIG 2}
\end{figure}

And the sets $\{D_n\}_1^4$ admit the following properties:
\[
\ba{l}
D_1=\{k\in\C|\re{l_1}>\re{l_2}=\re{l_3},\re{z_1}>\re{z_2}=\re{z_3}\},\\
D_2=\{k\in\C|\re{l_1}>\re{l_2}=\re{l_3},\re{z_1}<\re{z_2}=\re{z_3}\},\\
D_3=\{k\in\C|\re{l_1}<\re{l_2}=\re{l_3},\re{z_1}>\re{z_2}=\re{z_3}\},\\
D_4=\{k\in\C|\re{l_1}<\re{l_2}=\re{l_3},\re{z_1}<\re{z_2}=\re{z_3}\},\\
\ea
\]
where $l_i(\lambda)$ and $z_i(\lambda)$ are the diagonal entries of matrices $i\lambda\Lam$ and $2i\lambda^2\Lam$, respectively.
\par
%In fact, for $x=0$, $\mu_1(0,t,k)$ has enlarged domain of boundedness: $(D_2\cup D_4,D_1\cup D_3,D_1\cup D_3)$, and $\mu_2(0,t,k)$ has enlarged domain of boundedness: $(D_1\cup D_3,D_2\cup D_4,D_2\cup D_4)$.

\subsection{The spectral functions  $M_n$'s}
For each $n=1,\ldots,4$, a solution $M_n(x,t,\lambda)$ of (\ref{muLax}) can be defined by the following system of integral equations:
\be\label{Mndef}
(M_n)_{ij}(x,t,\lambda)=\dta_{ij}+\int_{\gam_{ij}^n}\left(e^{-(i \lambda x+2i \lambda^2t)\hat \Lam}W_n(x',t',\lambda)\right)_{ij},\quad \lambda\in D_n,\quad i,j=1,2,3.
\ee
where $W_n$ is given by (\ref{Wdef}) with $\mu$ replaced with $M_n$, and the contours $\gam_{ij}^n$, $n=1,\ldots,4$, $i,j=1,2,3$ are defined by
\be\label{gamijndef}
\gam_{ij}^n=\left\{\ba{lclcl}\gam_1&if&\re l_i(\lambda)<\re l_j(\lambda)&and&\re z_i(\lambda)\ge\re z_j(\lambda),\\
\gam_2&if&\re l_i(\lambda)<\re l_j(\lambda)&and&\re z_i(\lambda)<\re z_j(\lambda),\\
\gam_3&if&\re l_i(\lambda)\ge\re l_j(\lambda)&and &\re z_i(\lambda)\le \re z_j(\lambda),\\
\gam_4&if&\re l_i(\lambda)\ge\re l_j(\lambda)&and &\re z_i(\lambda)\ge \re z_j(\lambda).\ea\right.\quad \mbox{for }\quad \lambda \in D_n.
\ee
Here, we make a distinction between the contours $\gamma_3$ and $\gamma_4$ as follows,
\begin{equation}\label{gamijn34def}
 \gamma_{ij}^n =\left\{ \ba{lcl} \gamma_3, &if & \prod_{1\le i<j\le 3}(\re l_i(\lambda)-\re l_j(\lambda))(\re z_i(\lambda)-\re z_j(\lambda))<0,\\
 \gamma_4, &if & \prod_{1\le i<j\le 3}(\re l_i(\lambda)-\re l_j(\lambda))(\re z_i(\lambda)-\re z_j(\lambda))>0. \ea \right.
\end{equation}
The rule chosen in the produce is if $l_m=l_n$, $m$ may not equals $n$, we just choose the subscript is smaller one.

According to the definition of the $\gam^n$, one find that
\be\label{gamndef}
\ba{ll}
\gam^1=\left(\ba{lll}\gam_4&\gam_4&\gam_4\\\gam_2&\gam_4&\gam_4\\\gam_2&\gam_4&\gam_4\ea\right)&
\gam^2=\left(\ba{lll}\gam_3&\gam_3&\gam_3\\\gam_1&\gam_3&\gam_3\\\gam_1&\gam_3&\gam_3\ea\right)\\
\\
\gam^3=\left(\ba{lll}\gam_3&\gam_1&\gam_1\\\gam_3&\gam_3&\gam_3\\\gam_3&\gam_3&\gam_3\ea\right)&
\gam^4=\left(\ba{lll}\gam_4&\gam_2&\gam_2\\\gam_4&\gam_4&\gam_4\\\gam_4&\gam_4&\gam_4\ea\right).
\ea
\ee
\par
The following proposition ascertains that the $M_n$'s defined in this way have the properties required for the formulation of a Riemann-Hilbert problem.
\begin{proposition}
For each $n=1,\ldots,4$, the function $M_n(x,t,\lambda)$ is well-defined by equation (\ref{Mndef}) for $\lambda\in \bar D_n$ and $(x,t)\in \Om$.  Moreover, $M_n$ admits a bounded and contious extension to $\bar D_n$ and
\be\label{Mnasy}
M_n(x,t,\lambda)=\id+O(\frac{1}{\lambda}),\qquad \lambda\rightarrow \infty,\quad \lambda\in D_n.
\ee
\end{proposition}
\begin{proof}
Analogous to the proof provided in \cite{l1}
%Substituting the expansion
%\[
%M=M_0+\frac{M^{(1)}}{k}+\frac{M^{(2)}}{k^2}+\cdots,\qquad k\rightarrow \infty.
%\]
%into the Lax pair (\ref{muLaxe}) and comparing the terms of the same order of $k$ yield the equation (\ref{Mnasy}).
\end{proof}

\begin{remark}
Of course, for any fixed point $(x,t)$, $M_n$ is bounded and analytic as a function of $k\in D_n$ away from a possible discrete set of singularities $\{k_j\}$ at which the Fredholm determinant vanishes. The bounedness and analyticity properties are established in appendix B in \cite{l1}.
\end{remark}

\subsection{The jump matrices}

The spectral functions $\{S_n(\lambda)\}_1^4$ can be defined by
\be\label{Sndef}
S_n(\lambda)=M_n(0,0,\lambda),\qquad \lambda\in D_n,\quad n=1,\ldots,4.
\ee
Let $M$ denote the sectionally analytic function on the Riemann $\lambda-$plane which equals $M_n$ for $\lambda\in D_n$. Then $M$ satisfies the jump conditions
\be\label{Mjump}
M_n=M_mJ_{m,n},\qquad k\in \bar D_n\cap \bar D_m,\qquad n,m=1,\ldots,4,\quad n\ne m,
\ee
where the jump matrices $J_{m,n}(x,t,\lambda)$ are given by
\be\label{Jmndef}
J_{m,n}=e^{-(i \lambda x+2i \lambda^2t)\hat \Lam}(S_m^{-1}S_n).
\ee

\subsection{The adjugated eigenfunctions}
As the expressions of $S_n(\lambda)$ will involve the adjugate matrix of $\{s(\lambda), S(\lambda), S_L(\lambda) \}$ defined in the next subsection. We will also need the analyticity and boundedness  of the  the matrices $\{\mu_j(x,t,\lambda)\}_1^4$. We recall that the adjugate matrix $X^A$ of a $3\times 3$ matrix $X$ is defined by
\[
X^A=\left(
\ba{ccc}
m_{11}(X)&-m_{12}(X)&m_{13}(X)\\
-m_{21}(X)&m_{22}(X)&-m_{23}(X)\\
m_{31}(X)&-m_{32}(X)&m_{33}(X)
\ea
\right),
\]
where $m_{ij}(X)$ denote the $(ij)$th minor of $X$.
\par
It follows from (\ref{muLax}) that the adjugated eigenfunction $\mu^A$ satisfies the Lax pair
\be\label{muadgLaxe}
\left\{
\ba{l}
\mu_x^A-i\lambda[\Lam,\mu^A]=-V_1^T\mu^A,\\
\mu_t^A-2i\lambda^2[\Lam,\mu^A]=-V_2^T\mu^A.
\ea
\right.
\ee
where $V^T$ denotes the transform of a matrix $V$.
Thus, the eigenfunctions $\{\mu_j^A\}_1^4$ are solutions of the integral equations
\be\label{muadgdef}
\mu_j^A(x,t,\lambda)=\id-\int_{\gam_j}e^{i\lambda(x-x')+2i\lambda^2(t-t')\hat \Lam}(V_1^Tdx+V_2^Tdt)\mu_j^A,\quad j=1,2,3.
\ee
Then we can get the following analyticity and boundedness properties:
\be\label{mujadgbodanydom}
\ba{ll}
\mu_1^A:&(D_3,D_2,D_2),\\
\mu_2^A:&(D_4,D_1,D_1),\\
\mu_3^A:&(D_2,D_3,D_3),\\
\mu_4^A:&(D_1,D_4,D_4).
\ea
\ee
%In fact, for $x=0$, $\mu_1^A(0,t,k)$ has enlarged domain of boundedness: $(D_1\cup D_3,D_2\cup D_4,D_2\cup D_4)$, and $\mu_2^A(0,t,k)$ has enlarged domain of boundedness: $(D_2\cup D_4,D_1\cup D_3,D_1\cup D_3)$.

\subsection{Symmetries}
We will show that the eigenfunctions $\mu_j(x,t,k)$ satisfy an important symmetry.
\begin{proposition}
The eigenfunction $\psi(x,t,k)$ of the Lax pair (\ref{Lax}) satisfies the following symmetry:
\begin{equation}\label{symmetry}
\psi^{-1}(x,t,k)=\overline{\psi(x,t,\bar k)}^T=-\Lambda\psi(x,t,-k)\Lambda,
\end{equation}
here the superscript $T$ denotes a matrix transpose.
\begin{proof}
The matrices $U(x,t,k)$ and $V(x,t,k)$ in the Lax pair (\ref{Lax}) written in the form
\[
\psi_x=U\psi, \qquad \psi_t=V\psi,
\]
satisfy the following  symmetry relations
\begin{equation}\label{bark}
U(x,t,k)^T=-\overline{U(x,t,\bar k)}, \qquad V(x,t,k)^T=-\overline{V(x,t,\bar k)},
\end{equation}
and
\begin{equation}\label{-k}
U(x,t,k)=\Lambda U(x,t,-k)\Lambda, \qquad V(x,t,k)=\Lambda V(x,t,-k)\Lambda.
\end{equation}
In turn, relations (\ref{bark}) and (\ref{-k}) imply
\begin{equation}
\psi_x^A(x,t,k)=\overline{U(x,t,\bar k)}\psi^A(x,t,k), \qquad \psi_t^A(x,t,k)=\overline{V(x,t,\bar k)}\psi^A(x,t,k),
\end{equation}
and
\begin{equation}
\psi_x^A(x,t,k)=-\Lambda U^T(x,t,-k)\Lambda\psi^A(x,t,k), \qquad \psi_t^A(x,t,k)=-\Lambda V(x,t,-k)\Lambda \psi^A(x,t,k).
\end{equation}
\end{proof}
\end{proposition}

\begin{remark}
From proposition 2.3, one can show that the eigenfunctions $\mu_j(x,t,\lambda)$ of Lax pair equations (\ref{muLax}) satisfy the same symmetry.
\end{remark}

\subsection{The $J_{m,n}$'s computation}

Let us define the $3\times 3-$matrix value spectral functions $s(\lambda)$, $S(\lambda)$ and $S_L(\lambda)$ by
\begin{subequations}\label{sSdef}
\be\label{mu3mu2s}
\mu_3(x,t,\lambda)=\mu_2(x,t,\lambda)e^{-(i\lambda x+2i\lambda^2t)\hat \Lam}s(\lambda), %k\in (D_4\cup D_5\cup D_6,D_4\cup D_5\cup D_6,D_1\cup D_2\cup D_3),
\ee
\be\label{mu1mu2S}
\mu_1(x,t,\lambda)=\mu_2(x,t,\lambda)e^{-(i\lambda x+2i\lambda^2t)\hat \Lam}S(\lambda), %k\in \C.
\ee
\be\label{mu3mu4SL}
\mu_4(x,t,\lambda)=\mu_3(x,t,\lambda)e^{-(i\lambda (x-L)+2i\lambda^2t)\hat \Lam}S_L(\lambda), %k\in \C.
\ee
\end{subequations}
Thus,
\begin{subequations}
\be\label{smu3}
s(\lambda)=\mu_3(0,0,\lambda),
\ee
\be\label{Smu1}
S(\lambda)=\mu_1(0,0,\lambda)=e^{2i\lambda^2T\hat \Lambda}\mu_2^{-1}(0,T,\lambda),
\ee
\be\label{SLmu4}
S_L(\lambda)=\mu_4(L,0,\lambda)=e^{2i\lambda^2T\hat \Lambda}\mu_3^{-1}(L,T,\lambda),
\ee
\end{subequations}
And we can deduce from the properties of $\mu_j$ and $\mu_j^A$  that $\{s(\lambda), S(\lambda), S_L(\lambda)\}$ and $\{s^A(\lambda), S^A(\lambda), S_L^A(\lambda)\}$ have the following boundedness properties:
\[
\ba{ll}
s(\lambda):&(D_3\cup D_4,D_1\cup D_2,D_1\cup D_2),\\
S(\lambda):&(D_2\cup D_4,D_1\cup D_3,D_1\cup D_3),\\
S_L(\lambda):&(D_2\cup D_4,D_1\cup D_3,D_1\cup D_3),\\
s^A(\lambda):&(D_1\cup D_2,D_3\cup D_4,D_3\cup D_4),\\
S^A(\lambda):&(D_1\cup D_3,D_2\cup D_4,D_2\cup D_4),\\
S_L^A(\lambda):&(D_1\cup D_3,D_2\cup D_4,D_2\cup D_4).
\ea
\]
Moreover,
\be\label{MnSnrel}
M_n(x,t,\lambda)=\mu_2(x,t,\lambda)e^{-(i\lambda x+2i\lambda^2t)\hat\Lam}S_n(\lambda),\quad \lambda \in D_n.
\ee

\begin{proposition}
The $S_n$ can be expressed in terms of the entries of $s(\lambda),S(\lambda)$ and $S_L(\lambda)$ as follows:
\begin{subequations}\label{Sn}
\be
\ba{ll}
S_1=\left(\ba{ccc}\frac{1}{m_{11}(\mathcal{A})}&\mathcal{A}_{12}&\mathcal{A}_{13}\\0&\mathcal{A}_{22}
&\mathcal{A}_{23}\\0&\mathcal{A}_{32}&\mathcal{A}_{33}\ea\right),&
S_2=\left(\ba{ccc}\frac{S_{11}}{(S^Ts^A)_{11}}&s_{12}&s_{13}\\\frac{S_{21}}{(S^Ts^A)_{11}}&s_{22}
&s_{23}\\\frac{S_{31}}{(S^Ts^A)_{11}}&s_{32}&s_{33}\ea\right),\\
\ea
\ee
\be
\ba{l}
S_3=\left(\ba{ccc}s_{11}&\frac{m_{33}(s)m_{21}(S)-m_{23}(s)m_{31}(S)}{(s^TS^A)_{11}}&\frac{m_{32}(s)m_{21}(S)-m_{22}(s)m_{31}(S)}{(s^TS^A)_{11}}\\
s_{21}&\frac{m_{33}(s)m_{11}(S)-m_{13}(s)m_{31}(S)}{(s^TS^A)_{11}}&\frac{m_{32}(s)m_{11}(S)-m_{12}(s)m_{31}(S)}{(s^TS^A)_{11}}\\
s_{31}&\frac{m_{23}(s)m_{11}(S)-m_{13}(s)m_{21}(S)}{(s^TS^A)_{11}}&\frac{m_{22}(s)m_{11}(S)-m_{12}(s)m_{21}(S)}{(s^TS^A)_{11}}\ea\right),\\
S_4=\left(\ba{ccc}\mathcal{A}_{11}&0&0\\\mathcal{A}_{21}&\frac{m_{33}(\mathcal{A})}{\mathcal{A}_{11}}&\frac{m_{32}(\mathcal{A})}{\mathcal{A}_{11}}\\\mathcal{A}_{31}&\frac{m_{23}(\mathcal{A})}{\mathcal{A}_{11}}&\frac{m_{22}(\mathcal{A})}{\mathcal{A}_{11}}\ea\right).
\ea
\ee
\end{subequations}
where $\mathcal{A}=(\mathcal{A}_{ij})_{i,j=1}^{3}$ is a $3\times 3$ matrix, which is defined as
$
\mathcal{A}=s(\lambda)e^{-i\lambda L\hat \Lam}S_L(\lambda).
$
And the functions
\[
(S^Ts^A)_{11}=S_{11}m_{11}(s)-S_{21}m_{21}(s)+S_{31}m_{31}(s),
\]
\[
(s^TS^A)_{11}=s_{11}m_{11}(S)-s_{21}m_{21}(S)+s_{31}m_{31}(S).
\]
\end{proposition}
\begin{proof}
Firstly, we define $R_n(\lambda),T_n(\lambda)$ and $Q_n(\lambda)$ as follows:
\begin{subequations}\label{RnTnQn}
\be\label{Rn}
R_n(\lambda)=e^{2i\lambda^2T\hat \Lam}M_n(0,T,\lambda),
\ee
\be\label{Tn}
T_n(\lambda)=e^{i\lambda L\hat \Lam}M_n(L,0,\lambda),
\ee
\be\label{Qn}
Q_n(\lambda)=e^{(i\lambda L+2i\lambda^2T)\hat\Lam}M_n(L,T,\lambda).
\ee
\end{subequations}
Then, we have the following relations:
\be\label{MnRnSnTn}
\left\{
\ba{l}
M_n(x,t,\lambda)=\mu_1(x,t,\lambda)e^{(i\lambda x+2i\lambda^2t)\hat \Lam}R_n(\lambda),\\
M_n(x,t,\lambda)=\mu_2(x,t,\lambda)e^{(i\lambda x+2i\lambda^2t)\hat \Lam}S_n(\lambda),\\
M_n(x,t,\lambda)=\mu_3(x,t,\lambda)e^{(i\lambda x+2i\lambda^2t)\hat \Lam}T_n(\lambda),\\
M_n(x,t,\lambda)=\mu_4(x,t,\lambda)e^{(i\lambda x+2i\lambda^2t)\hat \Lam}Q_n(\lambda).
\ea
\right.
\ee

The relations (\ref{MnRnSnTn}) imply that
\be\label{sSRnSnTn}
\ba{l}
s(\lambda)=S_n(\lambda)T^{-1}_n(\lambda),\\
S(\lambda)=S_n(\lambda)R^{-1}_n(\lambda),\\
\mathcal{A}(\lambda)=S_n(\lambda)Q^{-1}_n(\lambda).
\ea
\ee
These equations constitute a matrix factorization problem which, given $\{s(\lambda),S(\lambda),S_L(\lambda)\}$ can be solved for the $\{R_n,S_n,T_n,Q_n\}$. Indeed, the integral equations (\ref{Mndef}) together with the definitions of $\{R_n,S_n,T_n,Q_n\}$ imply that
\be
\left\{
\ba{lll}
(R_n(\lambda))_{ij}=0&if&\gam_{ij}^n=\gam_1,\\
(S_n(\lambda))_{ij}=0&if&\gam_{ij}^n=\gam_2,\\
(T_n(\lambda))_{ij}=\dta_{ij}&if&\gam_{ij}^n=\gam_3,\\
(Q_n(\lambda))_{ij}=\dta_{ij}&if&\gam_{ij}^n=\gam_4.
\ea
\right.
\ee
It follows that (\ref{sSRnSnTn}) are 27 scalar equations for 27 unknowns. By computing the explicit solution of this algebraic system, we arrive at (\ref{Sn}).
\end{proof}

\begin{remark}
Due to our symmetry, see Lemma \ref{symmetry}, the representation of the functions $S_n(\lambda)$ can be become simple. It leads to  much more simple to compute the jump matrices $J_{m,n}(x,t,\lambda)$.
\end{remark}

\subsection{The residue conditions}
Since $\mu_2$ is an entire function, it follows from (\ref{MnSnrel}) that M can only have sigularities at the points where the $S_n's$ have singularities.
We denote the possible zeros by $\{\lambda_j\}_1^N$ and assume they satisfy the following assumption.
%\begin{assumption}\label{assum}
We assume that
\begin{itemize}
\item $m_{11}(\mathcal{A})(\lambda)$ has $n_0$ possible simple zeros in $D_1$ denoted by $\{\lambda_j\}_1^{n_0}$;
\item $(S^Ts^A)_{11}(k)$ has $n_1-n_0$ possible simple zeros in $D_2$ denoted by $\{\lambda_j\}_{n_0+1}^{n_1}$;
\item $(s^TS^A)_{11}(k)$ has $n_2-n_1$ possible simple zeros in $D_3$ denoted by $\{\lambda_j\}_{n_1+1}^{n_2}$;
\item $\mathcal{A}_{11}(k)$ has $N-n_2$ possible simple zeros in $D_4$ denoted by $\{\lambda_j\}_{n_2+1}^{N}$;
\end{itemize}
and that none of these zeros coincide. Moreover, we assume that none of these functions have zeros on the boundaries of the $D_n$'s.
%\end{assumption}
We determine the residue conditions at these zeros in the following:
\begin{proposition}\label{propos}
Let $\{M_n\}_1^4$ be the eigenfunctions defined by (\ref{Mndef}) and assume that the set $\{\lambda_j\}_1^N$ of singularities are as the above assumption. Then the following residue conditions hold:
\begin{subequations}
\be\label{M11D1res}
{Res}_{\lambda=\lambda_j}[M]_1=\frac{\mathcal{A}_{33}(\lambda_j)[M(\lambda_j)]_2-\mathcal{A}_{23}(\lambda_j)[M(\lambda_j)]_3}{\dot m_{11}(\mathcal{A})(\lambda_j)m_{21}(\mathcal{A})(\lambda_j)}e^{2\tha(\lambda_j)},\quad 1\le j\le n_0,\lambda_j\in D_1
\ee
\be\label{M21D2res}
\ba{r}
Res_{\lambda=\lambda_j}[M]_1=\frac{S_{21}(\lambda_j)s_{33}(\lambda_j)-S_{31}(\lambda_j)s_{23}(\lambda_j)}{\dot{(S^Ts^A)_{11}(\lambda_j)}m_{11}(\lambda_j)}e^{2\tha(\lambda_j)}[M(\lambda_j)]_2\\
{}+\frac{S_{31}(\lambda_j)s_{22}(\lambda_j)-S_{21}(\lambda_j)s_{32}(\lambda_j)}{\dot{(S^Ts^A)_{11}(\lambda_j)}m_{11}(\lambda_j)}e^{2\tha(\lambda_j)}[M(\lambda_j)]_3\\
\quad n_0+1\le j\le n_1,\lambda_j\in D_2,
\ea
\ee
\be\label{M32D3res}
\ba{r}
Res_{\lambda=\lambda_j}[M]_2=\frac{m_{33}(s)(\lambda_j)m_{21}(S)(\lambda_j)-m_{23}(s)(\lambda_j)m_{31}(S)(\lambda_j)}{\dot{(s^TS^A)_{11}(\lambda_j)}s_{11}(\lambda_j)}e^{-2\tha(\lambda_j)}[M(\lambda_j)]_1\\
\quad n_1+1\le j\le n_2,\lambda_j\in D_3,
\ea
\ee

\be\label{M33D3res}
\ba{r}
Res_{\lambda=\lambda_j}[M]_3=\frac{m_{32}(s)(\lambda_j)m_{21}(S)(\lambda_j)-m_{22}(s)(\lambda_j)m_{31}(S)(\lambda_j)}{\dot{(s^TS^A)_{11}(\lambda_j)}s_{11}(\lambda_j)}e^{-2\tha(\lambda_j)}[M(\lambda_j)]_1\\
\quad n_1+1\le j\le n_2,\lambda_j\in D_3.
\ea
\ee
\be\label{M42D4res}
Res_{\lambda=\lambda_j}[M]_2=\frac{m_{33}(\mathcal{A})(\lambda_j)}{\dot{\mathcal{A}}_{11}(\lambda_j) \mathcal{A}_{21}(\lambda_j)}e^{-2\tha(\lambda_j)}[M(\lambda_j)]_1,
\quad n_2+1\le j\le N,\lambda_j\in D_4.
\ee
\be\label{M43D4res}
Res_{\lambda=\lambda_j}[M]_3=\frac{m_{22}(\mathcal{A})(\lambda_j)}{\dot{\mathcal{A}}_{11}(\lambda_j) \mathcal{A}_{21}(\lambda_j)}e^{-2\tha(\lambda_j)}[M(\lambda_j)]_1,
\quad n_2+1\le j\le N,\lambda_j\in D_4.
\ee
\end{subequations}
where $\dot f=\frac{df}{d\lambda}$, and $\tha$ is defined by
\be\label{thaijdef}
\tha(x,t,\lambda)=i\lambda x+2i\lambda^2t.
\ee
\end{proposition}
\begin{proof}
We will prove (\ref{M11D1res}),  (\ref{M32D3res}),  the other conditions follow by similar arguments.
Equation (\ref{MnSnrel}) implies the relation
\begin{subequations}
\be\label{M1S1}
M_1=\mu_2e^{(i\lambda x+2i\lambda^2t)\hat\Lam}S_1,
\ee

\be\label{M3S3}
M_3=\mu_2e^{(i\lambda x+2i\lambda^2t)\hat\Lam}S_3,
\ee

\end{subequations}
In view of the expressions for $S_1$ and $S_3$ given in (\ref{Sn}), the three columns of (\ref{M1S1}) read:
\begin{subequations}
\be\label{M11}
[M_1]_1=[\mu_2]_1\frac{1}{m_{11}(\mathcal{A})},
\ee
\be\label{M12}
[M_1]_2=[\mu_2]_1e^{-2\tha}\mathcal{A}_{12}+[\mu_2]_2\mathcal{A}_{22}+[\mu_2]_3\mathcal{A}_{32},
\ee
\be\label{M13}
[M_1]_3=[\mu_2]_1e^{-2\tha}\mathcal{A}_{13}+[\mu_2]_2\mathcal{A}_{23}+[\mu_2]_3\mathcal{A}_{33}.
\ee
\end{subequations}
while the three columns of (\ref{M3S3}) read:
\begin{subequations}
\be\label{M31}
[M_3]_1=[\mu_2]_1s_{11}+[\mu_2]_2s_{21}e^{2\tha}+[\mu_2]_3s_{31}e^{2\tha}
\ee
\be\label{M32}
\ba{rl}
[M_3]_2&=[\mu_2]_1\frac{m_{33}(s)m_{21}(S)-m_{23}(s)m_{31}(S)}{(s^TS^A)_{11}}e^{-2\tha}\\
&+[\mu_2]_2\frac{m_{33}(s)m_{11}(S)-m_{13}(s)m_{31}(S)}{(s^TS^A)_{11}}\\
&+[\mu_2]_3\frac{m_{23}(s)m_{11}(S)-m_{13}(s)m_{21}(S)}{(s^TS^A)_{11}}
\ea
\ee
\be\label{M33}
\ba{rl}
[M_3]_3&=[\mu_2]_1\frac{m_{32}(s)m_{21}(S)-m_{22}(s)m_{31}(S)}{(s^TS^A)_{11}}e^{-2\tha}\\
&+[\mu_2]_2\frac{m_{32}(s)m_{11}(S)-m_{12}(s)m_{31}(S)}{(s^TS^A)_{11}}\\
&+[\mu_2]_3\frac{m_{22}(s)m_{11}(S)-m_{12}(s)m_{21}(S)}{(s^TS^A)_{11}}.
\ea
\ee
\end{subequations}

We first suppose that $\lambda_j\in D_1$ is a simple zero of $m_{11}(\mathcal{A})(\lambda)$. Solving (\ref{M12}) and (\ref{M13}) for $[\mu_2]_1,[\mu_2]_3$ and substituting the result in to (\ref{M11}), we find
\[
[M_1]_1=\frac{\mathcal{A}_{33}[M_1]_2-\mathcal{A}_{32}[M_1]_3}{m_{11}(\mathcal{A})m_{21}(\mathcal{A})}e^{2\tha}-\frac{[\mu_2]_2}{m_{21}(\mathcal{A})}e^{2\tha}.
\]
Taking the residue of this equation at $\lambda_j$, we find the condition (\ref{M11D1res}) in the case when $\lambda_j\in D_1$.
\par
In order to prove (\ref{M32D3res}), we solve (\ref{M31}) for $[\mu_2]_1$, then substituting the result into (\ref{M32}) and (\ref{M33}), we find
\begin{subequations}
\be
%\ba{l}
[M_3]_2=\frac{m_{33}(s)}{s_{11}}[\mu_2]_2+\frac{m_{23}(s)}{s_{11}}[\mu_2]_3+\frac{m_{33}(s)m_{21}(S)-m_{23}(s)m_{31}(S)}{\dot {(s^TS^A)_{11}}s_{11}}e^{-2\tha}[M_3]_1,
\ee
\be
[M_3]_3=\frac{m_{32}(s)}{s_{11}}[\mu_2]_2+\frac{m_{22}(s)}{s_{11}}[\mu_2]_3+\frac{m_{32}(s)m_{21}(S)-m_{22}(s)m_{31}(S)}{\dot {(s^TS^A)_{11}}s_{11}}e^{-2\tha}[M_3]_1.
\ee
\end{subequations}
Taking the residue of this equation at $\lambda_j$, we find the condition (\ref{M32D3res}) in the case when $\lambda_j\in D_3$. %Similarly, solving (\ref{M61}) and (\ref{M62}) for $[\mu_2]_1$ and $[\mu_2]_3$, then substituting the result into (\ref{M63}), we find
%\[
%[M_4]_3=\frac{s_{12}}{m_{33}(s)m_{23}(s)}e^{\tha_{13}}[M_4]_1-\frac{s_{11}}{m_{33}(s)m_{23}(s)}e^{\tha_{13}}[M_4]_2-\frac{1}{m_{23}(s)}e^{\tha_{23}}[\mu_2]_2.
%\]
%Taking the residue of this equation at $k_j$, we find the condition (\ref{M63D6res}) in the case when $k_j\in D_4$.
\end{proof}

\subsection{The global relation}
The spectral functions $S(\lambda),S_L(\lambda)$ and $s(\lambda)$ are not independent but satisfy an important relation. Indeed, it follows from (\ref{sSdef}) that
\be
\mu_1(x,t,\lambda)e^{-(i\lambda x+2i\lambda^2t)\hat \Lam}\{S^{-1}(\lambda)s(\lambda)e^{i\lambda L\hat\Lam}S_L(\lambda)\}=\mu_4(x,t,\lambda).%\quad k\in(D_3\cup D_4,D_3\cup D_4,D_1\cup D_2).
\ee
Since $\mu_1(0,T,\lambda)=\id$, evaluation at $(0,T)$ yields the following global relation:
\be\label{globalrel}
S^{-1}(\lambda)s(\lambda)e^{i\lambda L\hat\Lam}S_L(\lambda)=e^{2i\lambda^2T\hat \Lam}c(T,\lambda),%\quad k\in(D_3\cup D_4,D_3\cup D_4,D_1\cup D_2).
\ee
where $c(T,\lambda)=\mu_4(0,T,\lambda)$.

\section{The Riemann-Hilbert problem}

The sectionally analytic function $M(x,t,\lambda)$ defined in section 2 satisfies a Riemann-Hilbert problem which can be formulated in terms of the initial and boundary values of $q_1(x,t)$ and $q_2(x,t)$. By solving this Riemann-Hilbert problem, the solution of (\ref{cGI}) can be recovered for all values of $x,t$.
\begin{theorem}
Suppose that $q_1(x,t)$ and $q_2(x,t)$ are a pair of solutions of (\ref{cGI}) in the interval domain $\Om$. Then $q_1(x,t)$ and $q_2(x,t)$ can be reconstructed from the initial value $\{q_{10}(x),q_{20}(x)\}$ and boundary values $\{g_{01}(t),g_{02}(t),g_{11}(t),g_{12}(t)\}$, $\{f_{01}(t),f_{02}(t),f_{11}(t),f_{12}(t)\}$ defined as follows,
\be\label{inibouvalu}
\ba{ll}
q_{10}(x)=q_1(x,t=0),& q_{20}(x)=q_2(x,t=0),\\
g_{01}(t)=q_1(x=0,t),& g_{02}(t)=q_2(x=0,t),\\
f_{01}(t)=q_1(x=L,t),& f_{02}(t)=q_2(x=L,t),\\
g_{11}(t)=q_{1x}(x=0,t),& g_{12}(t)=q_{2x}(x=0,t),\\
f_{11}(t)=q_{1x}(x=L,t),& f_{12}(t)=q_{2x}(x=L,t).
\ea
\ee
\par
Use the initial and boundary data to define the jump matrices $J_{m,n}(x,t,\lambda)$ in terms of the spectral functions $s(\lambda)$ and $S(\lambda),S_L(\lambda)$ by equation (\ref{sSdef}).
\par
Assume that the possible zeros $\{\lambda_j\}_1^N$ of the functions $m_{11}(\mathcal{A})(\lambda)$, $(S^Ts^A)_{11}(\lambda)$, $(s^TS^A)_{11}(\lambda)$ and $\mathcal{A}_{11}(\lambda)$ are as the  assumption in subsection 2.8.
\par
Then the solution $\{q_1(x,t),q_2(x,t)\}$ is given by
\be\label{usolRHP}
q_1(x,t)=2i\lim_{\lambda \rightarrow \infty}(\lambda M(x,t,\lambda))_{12},\quad q_2(x,t)=2i\lim_{\lambda \rightarrow \infty}(\lambda M(x,t,\lambda))_{13} .
\ee
where $M(x,t,\lambda)$ satisfies the following $3\times 3$ matrix Riemann-Hilbert problem:
\begin{itemize}
\item $M$ is sectionally meromorphic on the Riemann $\lambda-$sphere with jumps across the contours $\bar D_n\cap \bar D_m,n,m=1,\cdots, 4$, see Figure 2.
\item Across the contours $\bar D_n\cap \bar D_m$, $M$ satisfies the jump condition
      \be\label{MRHP}
      M_n(x,t,\lambda)=M_m(x,t,\lambda)J_{m,n}(x,t,\lambda),\quad \lambda \in \bar D_n\cap \bar D_m,n,m=1,2,3,4.
      \ee
\item $M(x,t,\lambda)=\id+O(\frac{1}{\lambda}),\qquad \lambda \rightarrow \infty$.
\item The residue condition of $M$ is showed in Proposition \ref{propos}.
\end{itemize}
\end{theorem}
\begin{proof}
It only remains to prove (\ref{usolRHP}) and this equation follows from the large $\lambda$ asymptotics of the eigenfunctions.
\end{proof}

\section{Non-linearizable Boundary Conditions}
A key  difficulty of initial-boundary value problems is that some of the boundary values are unkown for a well-posed problem. While we need  all boundary values to define  the spectral functions $S(\lambda)$ and $S_L(\lambda)$, and hence for the formulation of the Riemann-Hilbert problem. Our main result, Theorem 4.3, expresses the unknown boundary data  in terms of the prescribed boundary data and the initial data in terms of  the solution of a system of nonlinear integral equations.

\subsection{Asymptotics}
An analysis of (\ref{muLax}) shows that the eigenfunctions $\{\mu_j\}_1^4$ have the following asymptotics as $\lambda \rightarrow\infty$:
\be\label{mujasykinf}
\ba{l}
\mu_j(x,t,\lambda)=\id+\frac{1}{\lambda}\left(\ba{lll}\mu^{(1)}_{11}&\mu^{(1)}_{12}&\mu^{(1)}_{13}\\\mu^{(1)}_{21}&\mu^{(1)}_{22}&\mu^{(1)}_{23}\\\mu^{(1)}_{31}&\mu^{(1)}_{32}&\mu^{(1)}_{33}\ea\right)
+\frac{1}{\lambda^2}\left(\ba{lll}\mu^{(2)}_{11}&\mu^{(2)}_{12}&\mu^{(2)}_{13}\\\mu^{(2)}_{21}&\mu^{(2)}_{22}&\mu^{(2)}_{23}\\\mu^{(2)}_{31}&\mu^{(2)}_{32}&\mu^{(2)}_{33}\ea\right)
+O(\frac{1}{\lambda^3})\\
=\id+\frac{1}{\lambda}\left(\ba{ccc}\int_{(x_j,t_j)}^{(x,t)} \Dta_{11}dx +\eta_{11}dt&\frac{1}{2i} q_1&\frac{1}{2i} q_2\\
-\frac{1}{4}\bar q_{1x}+\frac{q_1}{8i}|q|^2&\int_{(x_j,t_j)}^{(x,t)}\Dta_{22}dx+\eta_{22}dt&\int_{(x_j,t_j)}^{(x,t)}\Dta_{23}dx+\eta_{23}dt\\
-\frac{1}{4}\bar q_{2x}+\frac{q_2}{8i}|q|^2&\int_{(x_j,t_j)}^{(x,t)}\Dta_{32}dx+\eta_{32}dt&\int_{(x_j,t_j)}^{(x,t)}\Dta_{33}dx+\eta_{33}dt \ea\right)\\
+\frac{1}{\lambda^2}\left(\ba{ccc}\mu^{(2)}_{11}&\frac{1}{4}\bar q_{1x}+\frac{1}{2i} (q_1 \mu^{(1)}_{22}+q_2 \mu^{(1)}_{32}) &\frac{1}{4}\bar q_{2x}+\frac{1}{2i}(q_1 \mu^{(1)}_{23}+q_2 \mu^{(1)}_{33})\\\mu^{(2)}_{21}\mu^{(2)}_{22}&\mu^{(2)}_{23}\\\mu^{(2)}_{31}&\mu^{(2)}_{32}&\mu^{(2)}_{33}\ea\right)

+O(\frac{1}{\lambda^3}).
\ea
\ee
where
\begin{equation}
  |q|^2=|q_1|^2+|q_2|^2,
\end{equation}
\begin{subequations}
\be\label{Dtadef}
\ba{l}
\Dta_{11}=\frac{1}{8i}|q|^4-\frac{1}{4}( q_1 \bar q_{1x}+ q_2 \bar q_{2x}),\\
\Dta_{22}=\frac{i}{8}|q_1|^2|q|^2+\frac{1}{4} q_{1}\bar q_{1x},\\
\Dta_{23}=\frac{i}{8}\bar q_1q_2|q|^2+\frac{1}{4} q_{2}\bar q_{1x},\\
\Dta_{32}=\frac{i}{8}q_1\bar q_2|q|^2+\frac{1}{4}q_{1}\bar q_{2x},\\
\Dta_{33}=\frac{i}{8}|q_2|^2|q|^2+\frac{1}{4} q_{2}\bar q_{2x}.
\ea
\ee
%\begin{small}
\be\label{etadef}
\ba{l}
%\begin{split}
\eta_{11}=\frac{1}{8i}|q|^6+\frac{1}{4}(\bar q_1 q_{1x}+\bar q_2 q_{2x}-q_1 \bar q_{1x}-q_2\bar q_{2x}) |q|^2
+\frac{1}{4i}(|q_{1x}|^2+|q_{2x}|^2)-\frac{1}{4}(q_1\bar q_{1t}+ q_2\bar q_{2t}),\\[8pt]
%\end{split}
%\begin{split}
\eta_{22}=-\frac{1}{8i}|q_1|^2|q|^4+\frac{1}{8}|q|^2(q_1\bar{q_{1x}}-\bar{q_1}q_{1x})
+\frac{1}{8}(q_1\bar q_{1x}-\bar q_2q_{2x}+q_2\bar q_{2x}-\bar q_1 q_{1x})|q_1|^2+\frac{i}{4}|q_{1x}|^2+\frac{1}{4} q_1 \bar q_{1t},\\[8pt]
%\end{split}
%\begin{split}
\eta_{23}=-\frac{1}{8i}\bar q_1 q_2|q|^4+\frac{1}{8}|q|^2(q_2\bar q_{1x}-\bar{q_1}q_{2x})
+\frac{1}{8}(q_1\bar q_{1x}-\bar q_2q_{2x}+q_2\bar q_{2x}-\bar q_1 q_{1x})\bar q_1 q_2+\frac{i}{4}\bar{q_{1x}}q_{2x}+\frac{1}{4} q_{2}\bar q_{1t},\\[8pt]
%\end{split}
%\begin{split}
\eta_{32}=-\frac{1}{8i}q_1 \bar q_2|q|^4+\frac{1}{8}|q|^2(q_1\bar q_{2x}-\bar{q_2}q_{1x})
+\frac{1}{8}(q_1\bar q_{1x}-\bar q_2q_{2x}+q_2\bar q_{2x}-\bar q_1 q_{1x})q_1 \bar q_2+\frac{i}{4}q_{1x}\bar{q_{2x}}+\frac{1}{4}q_{1}\bar q_{2t},\\[8pt]
%\end{split}
%\begin{split}
\eta_{33}=-\frac{1}{8i}|q_2|^2|q|^4+\frac{1}{8}|q|^2(q_2\bar q_{2x}-\bar{q_2}q_{2x})
+\frac{1}{8}(q_1\bar q_{1x}-\bar q_2q_{2x}+q_2\bar q_{2x}-\bar q_1 q_{1x})|q_2|^2+\frac{i}{4}|q_{2x}|^2+\frac{1}{4} q_2\bar q_{2t}.
%\end{split}
\ea
\ee
%\end{small}
\end{subequations}

\begin{remark}
The explicit formulas of $\mu_{11}^{(2)}$ and $\mu_{ij}^{(2)}, i, j= 2, 3$ are not presented in the following analysis, we do not write down the asymptotic expressions of these functions.
\end{remark}

Next, we define functions $\{\Phi_{ij}(t,\lambda)\}_{i,j=1}^{3}$  and ${\phi_{ij}(t,\lambda)}_{i,j=1}^3$ by:
\be
\mu_2(0,t,\lambda)=\left(\ba{lll}\Phi_{11}(t,\lambda)&\Phi_{12}(t,\lambda)&\Phi_{13}(t,\lambda)\\
\Phi_{21}(t,\lambda)&\Phi_{22}(t,\lambda)&\Phi_{23}(t,\lambda)\\\Phi_{31}(t,\lambda)&\Phi_{32}(t,\lambda)&\Phi_{33}(t,\lambda)\ea\right),
\ee
\be
\mu_3(L,t,\lambda)=\left(\ba{lll}\phi_{11}(t,\lambda)&\phi_{12}(t,\lambda)&\phi_{13}(t,\lambda)\\
\phi_{21}(t,\lambda)&\phi_{22}(t,\lambda)&\phi_{23}(t,\lambda)\\ \phi_{31}(t,\lambda)&\phi_{32}(t,\lambda)&\phi_{33}(t,\lambda)\ea\right).
\ee

\par
From the asymptotic of $\mu_j(x,t,\lambda)$ in (\ref{mujasykinf}) we have
\be\label{mu2x0tk}
\ba{l}
\mu_{2}(0,t,\lambda)=\id+\frac{1}{\lambda}\left(\ba{ccc}\Phi_{11}^{(1)}(t)&\Phi_{12}^{(1)}(t)&\Phi_{13}^{(1)}(t)\\
\Phi_{21}^{(1)}(t)&\Phi_{22}^{(1)}(t)&\Phi_{23}^{(1)}(t)\\\Phi_{31}^{(1)}(t)&\Phi_{32}^{(1)}(t)&\Phi_{33}(t)^{(1)}\ea\right)\\
+\frac{1}{\lambda^2}\left(\ba{ccc}\Phi_{11}^{(2)}(t)&\Phi_{12}^{(2)}(t)&\Phi_{13}^{(2)}(t)\\
\Phi_{21}^{(2)}(t)&\Phi_{22}^{(2)}(t)&\Phi_{23}^{(2)}(t)\\\Phi_{31}^{(2)}(t)&\Phi_{32}^{(2)}(t)&\Phi_{33}^{(2)}(t)\ea\right)+O(\frac{1}{\lambda^3}).
\ea
\ee

\par
Recalling  the definition of the boundary data at $x=0$, we have
\begin{equation}
\ba{ll}
\Phi_{12}^{(1)}(t)=\frac{1}{2i} g_{01}(t),&\Phi_{12}^{(2)}(t)=\frac{1}{4} g_{11}+\frac{1}{2i}( g_{01}\Phi_{22}^{(1)}+g_{02}\Phi_{32}^{(1)}),\\
\Phi_{13}^{(1)}(t)=\frac{1}{2i} g_{02}(t),&\Phi_{13}^{(2)}(t)=\frac{1}{4} g_{12}+\frac{1}{2i}(g_{01}\Phi_{23}^{(1)}+ g_{02}\Phi_{33}^{(1)}).
%\Phi_{12}^{(1)}(t)=\frac{1}{4} g_{11}(t)+\frac{1}{8i}(|g_{01}|^2+|g_{02}|^2)g_{01},\\
%\Phi_{13}^{(1)}(t)=\frac{1}{4} g_{12}(t)+\frac{1}{8i}(|g_{01}|^2+|g_{02}|^2)g_{02}.
\ea
\end{equation}

In particular, we find the following expressions for the boundary values at $x=0$ :
\begin{subequations}\label{g}
\begin{equation}\label{g012}
g_{01}(t)=2i\Phi_{12}^{(1)}(t), \qquad g_{02}(t)=2i\Phi_{13}^{(1)}(t)
\end{equation}
\begin{equation}\label{g112}
\ba{c}
 g_{11}(t)=4\Phi_{12}^{(2)}(t)+2i( g_{01}(t)\Phi_{22}^{(1)}(t)+g_{02}\Phi_{32}^{(1)}(t)),\\
 g_{12}(t)=4\Phi_{13}^{(2)}(t)+2i(g_{01}\Phi_{23}^{(1)}(t)+g_{02}(t)\Phi_{33}^{(1)}(t)).
\ea
\end{equation}
\end{subequations}

\par
Similarly, we have the asymptotic formulas for $\mu_3(L,t,\lambda)={\phi_{ij}(t,\lambda)}_{i,j=1}^3$,
\be\label{mu3xLtk}
\ba{l}
\mu_{3}(L,t,\lambda)=\id+\frac{1}{\lambda}\left(\ba{ccc}\phi_{11}^{(1)}(t)&\phi_{12}^{(1)}(t)&\phi_{13}^{(1)}(t)\\
\phi_{21}^{(1)}(t)&\phi_{22}^{(1)}(t)&\phi_{23}^{(1)}(t)\\ \phi_{31}^{(1)}(t)&\phi_{32}^{(1)}(t)&\phi_{33}(t)^{(1)}\ea\right)\\
+\frac{1}{\lambda^2}\left(\ba{ccc}\phi_{11}^{(2)}(t)&\phi_{12}^{(2)}(t)&\phi_{13}^{(2)}(t)\\
\phi_{21}^{(2)}(t)&\phi_{22}^{(2)}(t)&\phi_{23}^{(2)}(t)\\ \phi_{31}^{(2)}(t)&\phi_{32}^{(2)}(t)&\phi_{33}^{(2)}(t)\ea\right)+O(\frac{1}{\lambda^3}).
\ea
\ee

\par
Recalling that the definition of the boundary data at $x=L$, we have
\begin{equation}
\ba{ll}
\phi_{12}^{(1)}(t)=\frac{1}{2i} f_{01}(t),&\phi_{12}^{(2)}(t)=\frac{1}{4} f_{11}+\frac{1}{2i}( f_{01}\phi_{22}^{(1)}+f_{02}\phi_{32}^{(1)}),\\
\phi_{13}^{(1)}(t)=\frac{1}{2i} f_{02}(t),&\phi_{13}^{(2)}(t)=\frac{1}{4} f_{12}+\frac{1}{2i}(f_{01}\phi_{23}^{(1)} +f_{02}\phi_{33}^{(1)}.
%\phi_{12}^{(1)}(t)=\frac{1}{4} f_{11}(t)+\frac{1}{8i}(|f_{01}|^2+|f_{02}|^2)f_{01},\\
%\phi_{13}^{(1)}(t)=\frac{1}{4} f_{12}(t)+\frac{1}{8i}(|f_{01}|^2+|f_{02}|^2)f_{02}.
\ea
\end{equation}

In particular, we find the following expressions for the boundary values at $x=L$ :
\begin{subequations}\label{f}
\begin{equation}\label{f012}
f_{01}(t)=2i\phi_{12}^{(1)}(t), \qquad f_{02}(t)=2i\phi_{13}^{(1)}(t)
\end{equation}
\begin{equation}\label{f112}
\ba{c}
f_{11}(t)=4\phi_{12}^{(2)}(t)+2i( f_{01}(t)\phi_{22}^{(1)}(t)+f_{02}\phi_{32}^{(1)}(t)),\\
f_{12}(t)=4\phi_{13}^{(2)}(t)+2i(f_{01}\phi_{23}^{(1)}+ f_{02}(t)\phi_{33}^{(1)}(t)).
\ea
\end{equation}
\end{subequations}

From the global relation (\ref{globalrel})and replacing $T$ by $t$, we find
\be\label{globalrelsec}
\mu_2(0,t,\lambda)e^{-2i\lambda^2 t\hat \Lam}\{s(\lambda)e^{i\lambda L\hat \Lam }S_{L}(\lambda)\}=c(t,\lambda),\quad \lambda \in(D_3\cup D_4,D_1\cup D_2,D_1\cup D_2).
\ee
%Then we can write the columns of the global relation, undering the matrix partitioned as (\ref{block}), as
%\begin{subequations}\label{globalrelsec}
%\be\label{globalrel1j}
%\Phi_{11}(t,k)s_{1j}(k)s^{-1}_{2\times 2}(k)e^{-4ik^4t}+\Phi_{1j}(t,k)=c_{1j}(t,k)s^{-1}_{2\times 2}(k),\quad k\in D_1\cup D_2,
%\ee
%\be\label{globalrel2by2}
%\Phi_{j1}(t,k)s_{1j}(k)s^{-1}_{2\times 2}(k)e^{-4ik^4t}+\Phi_{2\times 2}(t,k)=c_{2\times 2}(t,k)s^{-1}_{2\times 2}(k),\quad k\in D_1\cup D_2,
%\ee
%\end{subequations}
%\\
%The functions $c_{1j}(t,k),c_{2\times 2}(t,k)$ are analytic and bounded in $D_1\cup D_2$ away from the possible zeros of $s_{11}(k), (s^TS^A)_{11}$ and of order $O(\frac{1}{k^3})$ as $k\rightarrow \infty$.
%We will also need the asymptotic of $c_{1j}(t,k)$ and $c_{2\times2}(t,k)$.
\begin{lemma}
We assuming that the initial value and boundary value are compatible at $x=0$ and $x=L$ , then in the vanishing initial value case, the global relation (\ref{globalrelsec}) implies that the large $\lambda$ behavior of $c_{j1}(t,\lambda), j=2,3$  satisfy
\begin{subequations}
\be\label{c21largek}
\begin{split}
c_{21}(t,\lambda)&=\frac{\Phi_{21}^{(1)}(t)}{\lambda}+\frac{\Phi_{21}^{(2)}(t)+\Phi_{21}^{(1)}(t)\bar \phi_{11}^{(1)}(t)}{\lambda^2}+O(\frac{1}{\lambda^3})\\
&+\left[\frac{\bar \phi_{12}^{(1)}(t)}{\lambda}+\frac{\bar \phi_{12}^{(2)}(t)+\Phi_{22}^{(1)}(t)\bar \phi_{12}^{(1)}(t)+\Phi_{23}^{(1)}(t)\bar \phi_{13}^{(1)}(t)}{\lambda^2}+O(\frac{1}{\lambda^3})\right]e^{-2i\lambda L},\quad \lambda \rightarrow \infty,
\end{split}
\ee
\be\label{c31largek}
\begin{split}
c_{31}(t,\lambda)&=\frac{\Phi_{31}^{(1)}(t)}{\lambda}+\frac{\Phi_{31}^{(2)}(t)+\Phi_{31}^{(1)}(t)\bar \phi_{11}^{(1)}(t)}{\lambda^2}+O(\frac{1}{\lambda^3})\\
&+\left[\frac{\bar \phi_{13}^{(1)}(t)}{\lambda}+\frac{\bar \phi_{13}^{(2)}(t)+\Phi_{32}^{(1)}(t)\bar \phi_{12}^{(1)}(t)+\Phi_{33}^{(1)}(t)\bar \phi_{13}^{(1)}(t)}{\lambda^2}+O(\frac{1}{\lambda^3})\right]e^{-2i\lambda L},\quad \lambda \rightarrow \infty,
\end{split}
\ee
\end{subequations}
\end{lemma}

\begin{proof}

The global relation shows that under the assumption of vanishing initial value
\begin{subequations}
\begin{equation}\label{globalrel21}
c_{21}(t, \lambda)=\Phi_{21}(t, \lambda)\bar \phi_{11}(t, \bar \lambda)+\Phi_{22}(t, \lambda)\bar \phi_{12}(t, \bar\lambda)e^{-2i\lambda L}+ \Phi_{23}(t, \lambda)\bar \phi_{13}(t, \bar\lambda)e^{-2i\lambda L},
\end{equation}
\begin{equation}\label{globalrel31}
c_{31}(t, \lambda)=\Phi_{31}(t, \lambda)\bar \phi_{11}(t, \bar\lambda)+\Phi_{32}(t, \lambda)\bar \phi_{12}(t, \bar\lambda)e^{-2i\lambda L}+ \Phi_{33}(t, \lambda)\bar \phi_{13}(t, \bar\lambda)e^{-2i\lambda L},
\end{equation}
\end{subequations}

Recalling the equation
\be\label{Phit}
\mu_t+2i\lambda^2[\Lam,\mu]=V_2\mu.
\ee
\begin{subequations}\label{Phiteqn}
From the first column of the equation (\ref{Phit}) we get
\bee\label{Phi1t}
\left\{ \begin{array}{rl}
 \Phi_{11t}=&\!\!\!\! \left[\frac{1}{2}(g_{01}\bar g_{11}+g_{02}\bar g_{12})+\frac{i}{4}(|g_{01}|^2+|g_{02}|^2)^2\right]\Phi_{11}
 +(2\lambda g_{01}+ig_{11})\Phi_{21}+(2\lambda g_{02}+ig_{12})\Phi_{31},\vspace{0.08in}\\
 \\
 \Phi_{21t}=&\!\!\!\!4i\lambda^2\Phi_{21}+\left[-\frac{1}{2}\lambda (|g_{01}|^2+|g_{02}|^2)\bar g_{01}+i\lambda \bar g_{11}+\frac{i}{4}\left((2|g_{01}|^2+|g_{02}|^2)\bar g_{11}+\bar g_{01}g_{02}\bar g_{12}-\bar g_{01}^2g_{11}\right.\right.\vspace{0.08in}\\
&\left.\left.-\bar g_{01}\bar g_{02}g_{12}\right)-\frac{1}{4}\bar g_{01}(|g_{01}|^2+|g_{02}|^2)^2-\frac{1}{2i}\bar g_{01t}\right]\Phi_{11}-\frac{i}{4}(|g_{01}|^2+|g_{02}|^2)|g_{01}|^2\Phi_{21}\vspace{0.08in}\\
&-\frac{1}{2}g_{01}\bar g_{11}\Phi_{21}-\left(\frac{1}{2}g_{02}\bar g_{11}+\frac{i}{4}(|g_{01}|^2+|g_{02}|^2)\bar g_{01}g_{02}\right)\Phi_{31},\vspace{0.08in}\\
 \\
 \Phi_{31t}=&\!\!\!\!4i\lambda^2\Phi_{31}+\left[-\frac{1}{2}\lambda (|g_{01}|^2+|g_{02}|^2)\bar g_{02}+i\lambda \bar g_{12}+\frac{i}{4}\left((2|g_{02}|^2+|g_{01}|^2)\bar g_{12}+\bar g_{02}g_{01}\bar g_{11}-\bar g_{02}^2g_{12}\right.\right.\vspace{0.08in}\\
&\left.\left.-\bar g_{01}g_{02}g_{11}\right)-\frac{1}{4}\bar g_{02}(|g_{01}|^2+|g_{02}|^2)^2-\frac{1}{2i}\bar g_{02t}\right]\Phi_{11}-\frac{i}{4}(|g_{01}|^2+|g_{02}|^2)g_{01}\bar g_{02}\Phi_{21}\vspace{0.08in}\\
&-\frac{1}{2}g_{01}\bar g_{12}\Phi_{21}-\left(\frac{1}{2}g_{02}\bar g_{12}+\frac{i}{4}(|g_{01}|^2+|g_{02}|^2)|g_{02}|^2\right)\Phi_{31}.
\end{array}\right.
\ene

%\be\label{Phi1t}
%\left\{
%\ba{l}
%%\begin{split}
%\Phi_{11t}=\left[\frac{1}{2}(g_{01}\bar g_{11}+g_{02}\bar g_{12})+\frac{i}{4}(|g_{01}|^2+|g_{02}|^2)^2\right]\Phi_{11}+(2\lambda g_{01}+ig_{11})\Phi_{21}+(2\lambda g_{02}+ig_{12})\Phi_{31},\\
%\\
%%\end{split}
%%\begin{split}
%\Phi_{21t}-4i\lambda^2\Phi_{21}=\left[-\frac{1}{2}\lambda (|g_{01}|^2+|g_{02}|^2)\bar g_{01}+i\lambda \bar g_{11}+\frac{i}{4}\left((2|g_{01}|^2+|g_{02}|^2)\bar g_{11}+\bar g_{01}g_{02}\bar g_{12}-\bar g_{01}^2g_{11}\right.\right.\\
%\left.\left.-\bar g_{01}\bar g_{02}g_{12}\right)-\frac{1}{4}\bar g_{01}(|g_{01}|^2+|g_{02}|^2)^2-\frac{1}{2i}\bar g_{01t}\right]\Phi_{11}+\left(-\frac{1}{2}g_{01}\bar g_{11}-\frac{i}{4}(|g_{01}|^2+|g_{02}|^2)|g_{01}|^2\right)\Phi_{21}\\
%+\left(-\frac{1}{2}g_{02}\bar g_{11}-\frac{i}{4}(|g_{01}|^2+|g_{02}|^2)\bar g_{01}g_{02}\right)\Phi_{31},\\
%\\
%%\end{split}
%%\begin{split}
%\Phi_{31t}-4i\lambda^2\Phi_{31}=\left[-\frac{1}{2}\lambda (|g_{01}|^2+|g_{02}|^2)\bar g_{02}+i\lambda \bar g_{12}+\frac{i}{4}\left((2|g_{02}|^2+|g_{01}|^2)\bar g_{12}+\bar g_{02}g_{01}\bar g_{11}-\bar g_{01}g_{02}g_{11}\right.\right.\\
%\left.\left.-\bar g_{02}^2g_{12}\right)-\frac{1}{4}\bar g_{02}(|g_{01}|^2+|g_{02}|^2)^2-\frac{1}{2i}\bar g_{02t}\right]\Phi_{11}+\left(-\frac{1}{2}g_{01}\bar g_{12}-\frac{i}{4}(|g_{01}|^2+|g_{02}|^2)g_{01}\bar g_{02}\right)\Phi_{21}\\
%+\left(-\frac{1}{2}g_{02}\bar g_{12}-\frac{i}{4}(|g_{01}|^2+|g_{02}|^2)|g_{02}|^2\right)\Phi_{31},
%%\end{split}
%\ea
%\right.
%\ee

From the second column of the equation (\ref{Phit}) we get
\bee\label{Phi2t}
\left\{ \begin{array}{rl}
 \Phi_{12t}=&\!\!\!\! -4i\lambda^2\Phi_{12}+\left[\frac{1}{2}(g_{01}\bar g_{11}+g_{02}\bar g_{12})+\frac{i}{4}(|g_{01}|^2+|g_{02}|^2)^2\right]\Phi_{12}
 +(2\lambda g_{01}+ig_{11})\Phi_{22}\vspace{0.08in}\\
 &+(2\lambda g_{02}+ig_{12})\Phi_{32},\vspace{0.08in}\\
 \\
 \Phi_{22t}=&\!\!\!\!\left[-\frac{1}{2}\lambda (|g_{01}|^2+|g_{02}|^2)\bar g_{01}+i\lambda \bar g_{11}+\frac{i}{4}\left((2|g_{01}|^2+|g_{02}|^2)\bar g_{11}+\bar g_{01}g_{02}\bar g_{12}-\bar g_{01}^2g_{11}\right.\right.\vspace{0.08in}\\
&\left.\left.-\bar g_{01}\bar g_{02}g_{12}\right)-\frac{1}{4}\bar g_{01}(|g_{01}|^2+|g_{02}|^2)^2-\frac{1}{2i}\bar g_{01t}\right]\Phi_{12}-\frac{i}{4}(|g_{01}|^2+|g_{02}|^2)|g_{01}|^2\Phi_{22}\vspace{0.08in}\\
&-\frac{1}{2}g_{01}\bar g_{11}\Phi_{22}-\left(\frac{1}{2}g_{02}\bar g_{11}+\frac{i}{4}(|g_{01}|^2+|g_{02}|^2)\bar g_{01}g_{02}\right)\Phi_{32},\vspace{0.08in}\\
 \\
 \Phi_{32t}=&\!\!\!\!\left[-\frac{1}{2}\lambda (|g_{01}|^2+|g_{02}|^2)\bar g_{02}+i\lambda \bar g_{12}+\frac{i}{4}\left((2|g_{02}|^2+|g_{01}|^2)\bar g_{12}+\bar g_{02}g_{01}\bar g_{11}-\bar g_{02}^2g_{12}\right.\right.\vspace{0.08in}\\
&\left.\left.-\bar g_{01}g_{02}g_{11}\right)-\frac{1}{4}\bar g_{02}(|g_{01}|^2+|g_{02}|^2)^2-\frac{1}{2i}\bar g_{02t}\right]\Phi_{12}-\frac{i}{4}(|g_{01}|^2+|g_{02}|^2)g_{01}\bar g_{02}\Phi_{22}\vspace{0.08in}\\
&-\frac{1}{2}g_{01}\bar g_{12}\Phi_{22}-\left(\frac{1}{2}g_{02}\bar g_{12}+\frac{i}{4}(|g_{01}|^2+|g_{02}|^2)|g_{02}|^2\right)\Phi_{32}.
\end{array}\right.
\ene
From the third column of the equation (\ref{Phit}) we get
\bee\label{Phi3t}
\left\{ \begin{array}{rl}
 \Phi_{13t}=&\!\!\!\! -4i\lambda^2\Phi_{13}+\left[\frac{1}{2}(g_{01}\bar g_{11}+g_{02}\bar g_{12})+\frac{i}{4}(|g_{01}|^2+|g_{02}|^2)^2\right]\Phi_{13}
 +(2\lambda g_{01}+ig_{11})\Phi_{23}\vspace{0.08in}\\
 &+(2\lambda g_{02}+ig_{12})\Phi_{33},\vspace{0.08in}\\
 \\
 \Phi_{23t}=&\!\!\!\!\left[-\frac{1}{2}\lambda (|g_{01}|^2+|g_{02}|^2)\bar g_{01}+i\lambda \bar g_{11}+\frac{i}{4}\left((2|g_{01}|^2+|g_{02}|^2)\bar g_{11}+\bar g_{01}g_{02}\bar g_{12}-\bar g_{01}^2g_{11}\right.\right.\vspace{0.08in}\\
&\left.\left.-\bar g_{01}\bar g_{02}g_{12}\right)-\frac{1}{4}\bar g_{01}(|g_{01}|^2+|g_{02}|^2)^2-\frac{1}{2i}\bar g_{01t}\right]\Phi_{13}-\frac{i}{4}(|g_{01}|^2+|g_{02}|^2)|g_{01}|^2\Phi_{23}\vspace{0.08in}\\
&-\frac{1}{2}g_{01}\bar g_{11}\Phi_{23}-\left(\frac{1}{2}g_{02}\bar g_{11}+\frac{i}{4}(|g_{01}|^2+|g_{02}|^2)\bar g_{01}g_{02}\right)\Phi_{33},\vspace{0.08in}\\
 \\
 \Phi_{33t}=&\!\!\!\!\left[-\frac{1}{2}\lambda (|g_{01}|^2+|g_{02}|^2)\bar g_{02}+i\lambda \bar g_{12}+\frac{i}{4}\left((2|g_{02}|^2+|g_{01}|^2)\bar g_{12}+\bar g_{02}g_{01}\bar g_{11}-\bar g_{02}^2g_{12}\right.\right.\vspace{0.08in}\\
&\left.\left.-\bar g_{01}g_{02}g_{11}\right)-\frac{1}{4}\bar g_{02}(|g_{01}|^2+|g_{02}|^2)^2-\frac{1}{2i}\bar g_{02t}\right]\Phi_{13}-\frac{i}{4}(|g_{01}|^2+|g_{02}|^2)g_{01}\bar g_{02}\Phi_{23}\vspace{0.08in}\\
&-\frac{1}{2}g_{01}\bar g_{12}\Phi_{23}-\left(\frac{1}{2}g_{02}\bar g_{12}+\frac{i}{4}(|g_{01}|^2+|g_{02}|^2)|g_{02}|^2\right)\Phi_{33}.
\end{array}\right.
\ene
\end{subequations}
Suppose
\be\label{Phi1albt}
\left(\ba{l}\Phi_{11}\\\Phi_{21}\\\Phi_{31}\ea\right)=\left(\alpha_0(t)+\frac{\alpha_1(t)}{\lambda}+\frac{\alpha_2(t)}{\lambda^2}+\cdots\right)
+\left(\beta_0(t)+\frac{\beta_1(t)}{\lambda}+\frac{\beta_2(t)}{\lambda^2}+\cdots\right)e^{4i\lambda^2t},
\ee
where the coefficients $\alpha_j(t)$ and $\beta_{j}(t)$, $j=0,1,2,\cdots$, are independent of $k$ and are $3\times 1$ matrix functions.
\par
To determine these coefficients,we substitute the above equation into equation (\ref{Phi1t}) and use the initial conditions
\[
\alpha_0(0)+\beta_0(0)=(\ba{ccc}1&0&0\ea)^T,\quad \alpha_1(0)+\beta_1(0)=(\ba{ccc}0&0&0\ea)^T.
\]
Then we get
\be\label{Phi2albtrsu}
\ba{l}
\left(\ba{l}\Phi_{11}\\\Phi_{21}\\\Phi_{31}\ea\right)=
\left(\ba{l}1\\0\\0\ea\right)+\frac{1}{\lambda}\left(\ba{l}\Phi_{11}^{(1)}\\\Phi_{21}^{(1)}\\\Phi_{31}^{(1)}\ea\right)+\frac{1}{\lambda^2}\left(\ba{l}\Phi_{11}^{(2)}\\\Phi_{21}^{(2)}\\\Phi_{31}^{(2)}\ea\right)+O(\frac{1}{\lambda^3})\\
{}+\left[\frac{1}{\lambda}\left(\ba{c}0\\-\Phi^{(1)}_{21}(0)\\-\Phi^{(1)}_{31}(0)\ea\right)+O(\frac{1}{\lambda^2})\right]e^{4i\lambda^2t}
\ea
\ee
\par
Similarly, suppose
\be\label{Phi2albt}
\left(\ba{l}\Phi_{12}\\\Phi_{22}\\\Phi_{32}\ea\right)=\left(\alpha_0(t)+\frac{\alpha_1(t)}{\lambda}+\frac{\alpha_2(t)}{\lambda^2}+\cdots\right)
+\left(\beta_0(t)+\frac{\beta_1(t)}{\lambda}+\frac{\beta_2(t)}{\lambda^2}+\cdots\right)e^{-4i\lambda^2t},
\ee
where the coefficients $\alpha_j(t)$ and $\beta_{j}(t)$, $j=0,1,2,\cdots$, are independent of $k$ and are $3\times 1$ matrix functions.
\par
To determine these coefficients,we substitute the above equation into equation (\ref{Phi2t}) and use the initial conditions
\[
\alpha_0(0)+\beta_0(0)=(\ba{ccc}0&1&0\ea)^T,\quad \alpha_1(0)+\beta_1(0)=(\ba{ccc}0&0&0\ea)^T.
\]
Then we get
\be\label{Phi2albtrsu}
\ba{l}
\left(\ba{l}\Phi_{12}\\\Phi_{22}\\\Phi_{32}\ea\right)=
\left(\ba{l}0\\1\\0\ea\right)+\frac{1}{\lambda}\left(\ba{l}\Phi_{12}^{(1)}\\\Phi_{22}^{(1)}\\\Phi_{32}^{(1)}\ea\right)+\frac{1}{\lambda^2}\left(\ba{l}\Phi_{12}^{(2)}\\\Phi_{22}^{(2)}\\\Phi_{32}^{(2)}\ea\right)+O(\frac{1}{\lambda^3})\\
{}+\left[\frac{1}{\lambda}\left(\ba{c}-\Phi^{(1)}_{12}(0)\\0\\0\ea\right)+
\frac{1}{\lambda^2}\left(\ba{c}-\Phi^{(2)}_{12}(0)+\Phi^{(1)}_{12}(0)\Phi^{(1)}_{22}+\Phi^{(1)}_{12}(0)\Phi^{(1)}_{32}\\ \frac{1}{4i}\left(-\frac{1}{2}(|g_{01}|^2+|g_{02}|^2)\bar g _{01}+i\bar g_{11}\right)\Phi^{(1)}_{12}(0)\\\frac{1}{4i}\left(-\frac{1}{2}(|g_{01}|^2+|g_{02}|^2)\bar g _{02}+i\bar g_{12}\right)\Phi^{(1)}_{12}(0)\ea\right)+O(\frac{1}{\lambda^2})\right]e^{-4i\lambda^2t}
\ea
\ee

Similar to the derivation of $\Phi_{i2},i=1,2,3$, from (\ref{Phi3t}) we can get the asymptotic formulas of $\Phi_{i3},i=1,2,3$
\be\label{Phi3albtrsu}
\ba{l}
\left(\ba{l}\Phi_{13}\\\Phi_{23}\\\Phi_{33}\ea\right)=
\left(\ba{l}0\\0\\1\ea\right)+\frac{1}{\lambda}\left(\ba{l}\Phi_{13}^{(1)}\\\Phi_{23}^{(1)}\\\Phi_{33}^{(1)}\ea\right)+
\frac{1}{\lambda^2}\left(\ba{l}\Phi_{13}^{(2)}\\\Phi_{23}^{(2)}\\\Phi_{33}^{(2)}\ea\right)+O(\frac{1}{\lambda^3})\\
{}+\left[\frac{1}{\lambda}\left(\ba{c}-\Phi^{(1)}_{13}(0)\\0\\0\ea\right)+
\frac{1}{\lambda^2}\left(\ba{c}-\Phi^{(2)}_{13}(0)+\Phi^{(1)}_{13}(0)\Phi^{(1)}_{23}+\Phi^{(1)}_{13}(0)\Phi^{(1)}_{33}\\ \frac{1}{4i}\left(-\frac{1}{2}(|g_{01}|^2+|g_{02}|^2)\bar g _{01}+i\bar g_{11}\right)\Phi^{(1)}_{13}(0)\\\frac{1}{4i}\left(-\frac{1}{2}(|g_{01}|^2+|g_{02}|^2)\bar g _{02}+i\bar g_{12}\right)\Phi^{(1)}_{13}(0)\ea\right)+O(\frac{1}{\lambda^2})\right]e^{-4i\lambda^2t}
\ea
\ee

Similar to (\ref{Phiteqn}), we also know that  $\{\phi_{ij}\}_{i,j=1}^{3}$ satisfy the similar partial derivative equations. Substituting these formulas into the equation (\ref{globalrel21}) and noticing that we assume that the initial value and boundary value are compatible at $x=0$ and $x=L$, we get the asymptotic behavior (\ref{c21largek}) of $c_{j1}(t,\lambda)$ as $\lambda \rightarrow \infty$. Similar to prove the formula (\ref{c31largek}).

\end{proof}

\subsection{The Dirichlet and Neumann problems}
In what follows, we can derive the effective characterizations of spectral function $S(\lambda), S_L(\lambda)$ for the Dirichlet ($\{g_{01}(t), g_{02}(t)\}$ and $\{f_{01}(t), f_{02}(t)\}$  prescribed), the Neumann ($\{g_{11}(t), g_{12}(t)\}$ and $\{f_{11}(t), f_{12}(t)\}$ prescribed) problems.
\par
Define the following new functions as

\be\label{f+-}
f_-(t,\lambda)=f(t,\lambda)-f(t,-\lambda),\quad  f_+(t,\lambda)=f(t,\lambda)+f(t,-\lambda),
\ee
Introducing
\be\label{DelSig}
\Delta(k)=e^{2i\lambda L}-e^{-2i\lambda L},\quad \Sigma (k)=e^{2i\lambda L}+e^{-2i\lambda L}
\ee
Denoting $\partial D_3^0$ as the boundary contour which is not included the zeros of $\Delta(\lambda)$.
\begin{theorem}\label{maintheom}
Let $T<\infty$. Let $q_0(x)=(q_{10}(x),q_{20}(x)),0\leq x\ge L$, be two initial functions.
\par
For the Dirichlet problem it is assumed that the function $\{ g_{01}(t), g_{02}(t)\},0\le t<T$, has sufficient smoothness and is compatible with $\{q_{10}(x), q_{20}(x)$ at $x=t=0$, that is
\[
q_{10}(0)=g_{01}(0), \quad q_{20}(0)=g_{02}(0).
\]
the function $\{f_{01}(t), f){02}(t)\}, 0\leq t < T$, has sufficient smoothness and is compatible with $q_{10}(x), q_{20}(x)$ at $x=L$, that is
\[
q_{10}(L)=f_{01}(0), \quad q_{20}(L)=f_{02}(0).
\]

\par
For the Neumann problem it is assumed that the functions $\{g_{11}(t), g{12}(t)\},0\le t<T$, has sufficient smoothness and is compatible with $q_0(x)$ at $x=t=0$. The functions $\{f_{11}(t), f_{12}(t)\},0\le t<T$, has sufficient smoothness and is compatible with $q_0(x)$ at $x=L$.

\par
%Suppose that $s_{11}(k)$ has a finite number of simple zeros in $D_1$.
\par
Then the spectral function $S(\lambda), S_L(\lambda)$ is given by
\be\label{Sk}
S(\lambda)=\left(\ba{ccc}\ol{\Phi_{11}(\bar \lambda)}&e^{4i\lambda^2T}\ol{\Phi_{21}(\bar \lambda)}&e^{4i\lambda^2T}\ol{\Phi_{31}(\bar \lambda)}\\e^{-4i\lambda^2T}\ol{\Phi_{12}(\bar \lambda)}&\ol{\Phi_{22}(\bar \lambda)}&\ol{\Phi_{32}(\bar \lambda)}\\e^{-4i\lambda^2T}\ol{\Phi_{13}(\bar \lambda)}&\ol{\Phi_{23}(\bar \lambda)}&\ol{\Phi_{33}(\bar \lambda)}\ea\right)
\ee

\be\label{SLk}
S_L(\lambda)=\left(\ba{ccc}\ol{\phi_{11}(\bar \lambda)}&e^{4i\lambda^2T}\ol{\phi_{21}(\bar \lambda)}&e^{4i\lambda^2T}\ol{\phi_{31}(\bar \lambda)}\\e^{-4i\lambda^2T}\ol{\phi_{12}(\bar \lambda)}&\ol{\phi_{22}(\bar \lambda)}&\ol{\phi_{32}(\bar \lambda)}\\e^{-4i\lambda^2T}\ol{\phi_{13}(\bar \lambda)}&\ol{\phi_{23}(\bar \lambda)}&\ol{\phi_{33}(\bar \lambda)}\ea\right)
\ee

and the complex-value functions $\{\Phi_{l3}(t,\lambda)\}_{l=1}^{3}$ satisfy the following system of integral equations:
\bee\label{Phil3sys}
\left\{ \begin{array}{rl}
 \Phi_{13}(t, \lambda)=&\!\!\!\! \int_0^te^{-4i\lambda^2(t-t')}\left\{\left[\frac{1}{2}(g_{01}\bar g_{11}+g_{02}\bar g_{12})+\frac{i}{4}(|g_{01}|^2+|g_{02}|^2)^2\right]\Phi_{13}
 +(2\lambda g_{01}+ig_{11})\Phi_{23}\right.\vspace{0.08in}\\
 &\left.+(2\lambda g_{02}+ig_{12})\Phi_{33}\right\}(t',\lambda)dt',\vspace{0.08in}\\
 \\
 \Phi_{23}(t, \lambda)=&\!\!\!\!\int_0^t\left\{\left[-\frac{1}{2}\lambda (|g_{01}|^2+|g_{02}|^2)\bar g_{01}+i\lambda \bar g_{11}+\frac{i}{4}\left((2|g_{01}|^2+|g_{02}|^2)\bar g_{11}+\bar g_{01}g_{02}\bar g_{12}\right.\right.\right.\vspace{0.08in}\\
&\left.\left.\left.-\bar g_{01}^2g_{11}-\bar g_{01}\bar g_{02}g_{12}\right)-\frac{1}{4}\bar g_{01}(|g_{01}|^2+|g_{02}|^2)^2-\frac{1}{2i}\bar g_{01t}\right]\Phi_{13}-\frac{1}{2}g_{01}\bar g_{11}\Phi_{23}\right.\vspace{0.08in}\\
&\left.-\frac{i}{4}(|g_{01}|^2+|g_{02}|^2)|g_{01}|^2\Phi_{23}-\left(\frac{1}{2}g_{02}\bar g_{11}+\frac{i}{4}(|g_{01}|^2+|g_{02}|^2)\bar g_{01}g_{02}\right)\Phi_{33}\right\}(t',\lambda)dt',\vspace{0.08in}\\
 \\
 \Phi_{33}(t, \lambda)=&\!\!\!\!1+\int_0^t\left\{\left[-\frac{1}{2}\lambda (|g_{01}|^2+|g_{02}|^2)\bar g_{02}+i\lambda \bar g_{12}+\frac{i}{4}\left((2|g_{02}|^2+|g_{01}|^2)\bar g_{12}+\bar g_{02}g_{01}\bar g_{11}\right.\right.\right.\vspace{0.08in}\\
&\left.\left.\left.-\bar g_{02}^2g_{12}-\bar g_{01}g_{02}g_{11}\right)-\frac{1}{4}\bar g_{02}(|g_{01}|^2+|g_{02}|^2)^2-\frac{1}{2i}\bar g_{02t}\right]\Phi_{13}-\frac{1}{2}g_{01}\bar g_{12}\Phi_{23}\right.\vspace{0.08in}\\
&\left.-\frac{i}{4}(|g_{01}|^2+|g_{02}|^2)g_{01}\bar g_{02}\Phi_{23}-\left(\frac{1}{2}g_{02}\bar g_{12}+\frac{i}{4}(|g_{01}|^2+|g_{02}|^2)|g_{02}|^2\right)\Phi_{33}\right\}(t',\lambda)dt'.
\end{array}\right.
\ene

and $\{\Phi_{l1}(t,\lambda)\}_{l=1}^{3},\{\Phi_{l2}(t,\lambda)\}_{l=1}^{3}$ satisfy the following system of integral equations:

\bee\label{Phil1sys}
\left\{ \begin{array}{rl}
 \Phi_{11}(t, \lambda)=&\!\!\!\! 1+\int_0^t\left\{\left[\frac{1}{2}(g_{01}\bar g_{11}+g_{02}\bar g_{12})+\frac{i}{4}(|g_{01}|^2+|g_{02}|^2)^2\right]\Phi_{11}
 +(2\lambda g_{01}+ig_{11})\Phi_{21}\right.\vspace{0.08in}\\
 &\left.+(2\lambda g_{02}+ig_{12})\Phi_{31}\right\}(t',\lambda)dt',\vspace{0.08in}\\
 \\
 \Phi_{21}(t, \lambda)=&\!\!\!\!\int_0^te^{4i\lambda^2(t-t')}\left\{\left[-\frac{1}{2}\lambda (|g_{01}|^2+|g_{02}|^2)\bar g_{01}+i\lambda \bar g_{11}+\frac{i}{4}\left((2|g_{01}|^2+|g_{02}|^2)\bar g_{11}+\bar g_{01}g_{02}\bar g_{12}\right.\right.\right.\vspace{0.08in}\\
&\left.\left.\left.-\bar g_{01}^2g_{11}-\bar g_{01}\bar g_{02}g_{12}\right)-\frac{1}{4}\bar g_{01}(|g_{01}|^2+|g_{02}|^2)^2-\frac{1}{2i}\bar g_{01t}\right]\Phi_{11}-\frac{1}{2}g_{01}\bar g_{11}\Phi_{21}\right.\vspace{0.08in}\\
&\left.-\frac{i}{4}(|g_{01}|^2+|g_{02}|^2)|g_{01}|^2\Phi_{21}-\left(\frac{1}{2}g_{02}\bar g_{11}+\frac{i}{4}(|g_{01}|^2+|g_{02}|^2)\bar g_{01}g_{02}\right)\Phi_{31}\right\}(t',\lambda)dt',\vspace{0.08in}\\
 \\
 \Phi_{31}(t, \lambda)=&\!\!\!\!\int_0^te^{4i\lambda^2(t-t')}\left\{\left[-\frac{1}{2}\lambda (|g_{01}|^2+|g_{02}|^2)\bar g_{02}+i\lambda \bar g_{12}+\frac{i}{4}\left((2|g_{02}|^2+|g_{01}|^2)\bar g_{12}+\bar g_{02}g_{01}\bar g_{11}\right.\right.\right.\vspace{0.08in}\\
&\left.\left.\left.-\bar g_{02}^2g_{12}-\bar g_{01}g_{02}g_{11}\right)-\frac{1}{4}\bar g_{02}(|g_{01}|^2+|g_{02}|^2)^2-\frac{1}{2i}\bar g_{02t}\right]\Phi_{11}-\frac{1}{2}g_{01}\bar g_{12}\Phi_{21}\right.\vspace{0.08in}\\
&\left.-\frac{i}{4}(|g_{01}|^2+|g_{02}|^2)g_{01}\bar g_{02}\Phi_{21}-\left(\frac{1}{2}g_{02}\bar g_{12}+\frac{i}{4}(|g_{01}|^2+|g_{02}|^2)|g_{02}|^2\right)\Phi_{31}\right\}(t',\lambda)dt'.
\end{array}\right.
\ene

\bee\label{Phil2sys}
\left\{ \begin{array}{rl}
 \Phi_{12}(t, \lambda)=&\!\!\!\! \int_0^te^{-4i\lambda^2(t-t')}\left\{\left[\frac{1}{2}(g_{01}\bar g_{11}+g_{02}\bar g_{12})+\frac{i}{4}(|g_{01}|^2+|g_{02}|^2)^2\right]\Phi_{12}
 +(2\lambda g_{01}+ig_{11})\Phi_{22}\right.\vspace{0.08in}\\
 &\left.+(2\lambda g_{02}+ig_{12})\Phi_{32}\right\}(t',\lambda)dt',\vspace{0.08in}\\
 \\
 \Phi_{22}(t, \lambda)=&\!\!\!\!1+\int_0^t\left\{\left[-\frac{1}{2}\lambda (|g_{01}|^2+|g_{02}|^2)\bar g_{01}+i\lambda \bar g_{11}+\frac{i}{4}\left((2|g_{01}|^2+|g_{02}|^2)\bar g_{11}+\bar g_{01}g_{02}\bar g_{12}\right.\right.\right.\vspace{0.08in}\\
&\left.\left.\left.-\bar g_{01}^2g_{11}-\bar g_{01}\bar g_{02}g_{12}\right)-\frac{1}{4}\bar g_{01}(|g_{01}|^2+|g_{02}|^2)^2-\frac{1}{2i}\bar g_{01t}\right]\Phi_{12}-\frac{1}{2}g_{01}\bar g_{11}\Phi_{22}\right.\vspace{0.08in}\\
&\left.-\frac{i}{4}(|g_{01}|^2+|g_{02}|^2)|g_{01}|^2\Phi_{22}-\left(\frac{1}{2}g_{02}\bar g_{11}+\frac{i}{4}(|g_{01}|^2+|g_{02}|^2)\bar g_{01}g_{02}\right)\Phi_{32}\right\}(t',\lambda)dt',\vspace{0.08in}\\
 \\
 \Phi_{32}(t, \lambda)=&\!\!\!\!1+\int_0^t\left\{\left[-\frac{1}{2}\lambda (|g_{01}|^2+|g_{02}|^2)\bar g_{02}+i\lambda \bar g_{12}+\frac{i}{4}\left((2|g_{02}|^2+|g_{01}|^2)\bar g_{12}+\bar g_{02}g_{01}\bar g_{11}\right.\right.\right.\vspace{0.08in}\\
&\left.\left.\left.-\bar g_{02}^2g_{12}-\bar g_{01}g_{02}g_{11}\right)-\frac{1}{4}\bar g_{02}(|g_{01}|^2+|g_{02}|^2)^2-\frac{1}{2i}\bar g_{02t}\right]\Phi_{12}-\frac{1}{2}g_{01}\bar g_{12}\Phi_{22}\right.\vspace{0.08in}\\
&\left.-\frac{i}{4}(|g_{01}|^2+|g_{02}|^2)g_{01}\bar g_{02}\Phi_{22}-\left(\frac{1}{2}g_{02}\bar g_{12}+\frac{i}{4}(|g_{01}|^2+|g_{02}|^2)|g_{02}|^2\right)\Phi_{32}\right\}(t',\lambda)dt'.
\end{array}\right.
\ene

\begin{enumerate}
\item For the Dirichlet problem, the unknown Neumann boundary value $\{g_{11}(t), g_{12}(t)$ and $\{f_{11}(t), f_{12}(t)$ are given by
\begin{subequations}\label{DtoNg}
\be\label{DtoNg11}
\begin{split}
g_{11}(t)=&\frac{2}{i\pi}\int_{\partial D_3}\frac{\Sigma}{\Delta}(\lambda\Phi_{12-}+i g_{01})d\lambda -\frac{1}{\pi} \int_{\partial D_3}( g_{01} \bar \phi_{22-}+g_{02} \bar \phi_{23-})d\lambda\\
&+\frac{4}{i\pi}\int_{\partial D_3}\frac{1}{\Delta}(\lambda \bar \phi_{21-}-2\bar \phi_{21}^{(1)})d\lambda +\frac{2}{\pi}\int_{\partial D_3}(g_{01}\Phi_{22-}+g_{02}\Phi_{32-})d\lambda\\
&+\frac{4}{i\pi}\int_{\partial D_3}\frac{\lambda}{\Delta}[(\Phi_{11}-1)\bar \phi_{21}+\Phi_{12}(\bar \phi_{22}-1)e^{-2i\lambda L}+\Phi_{13}\bar \phi_{23}e^{-2i\lambda L}]_-d\lambda,
\end{split}
\ee

\be\label{DtoNg12}
\begin{split}
g_{12}(t)=&\frac{2}{i\pi}\int_{\partial D_3}\frac{\Sigma}{\Delta}(\lambda\Phi_{13-}+i g_{02})d\lambda -\frac{1}{\pi} \int_{\partial D_3}( g_{01} \bar \phi_{32-}+g_{02} \bar \phi_{33-})d\lambda\\
&+\frac{4}{i\pi}\int_{\partial D_3}\frac{1}{\Delta}(\lambda \bar \phi_{31-}-2\bar \phi_{31}^{(1)})d\lambda +\frac{2}{\pi}\int_{\partial D_3}(g_{01}\Phi_{23-}+g_{02}\Phi_{33-})d\lambda\\
&+\frac{4}{i\pi}\int_{\partial D_3}\frac{\lambda}{\Delta}[(\Phi_{11}-1)\bar \phi_{31}+\Phi_{12}\bar \phi_{32}e^{-2i\lambda L}+\Phi_{13}(\bar \phi_{33}-1)e^{-2i\lambda L}]_-d\lambda,
\end{split}
\ee
\end{subequations}
and
\begin{subequations}\label{DtoNf}
\be\label{DtoNf11}
\begin{split}
f_{11}(t)=&-\frac{2}{i\pi}\int_{\partial D_3}\frac{\Sigma}{\Delta}(\lambda\phi_{12-}+i f_{01})d\lambda -\frac{1}{\pi} \int_{\partial D_3}( f_{01} \bar \Phi_{22-}+f_{02} \bar \Phi_{23-})d\lambda\\
&-\frac{4}{i\pi}\int_{\partial D_3}\frac{1}{\Delta}(\lambda \bar \Phi_{21-}-2\bar \Phi_{21}^{(1)})d\lambda +\frac{2}{\pi}\int_{\partial D_3}(f_{01}\phi_{22-}+f_{02}\phi_{32-})d\lambda\\
&-\frac{4}{i\pi}\int_{\partial D_3}\frac{\lambda}{\Delta}[(\phi_{11}-1)\bar \Phi_{21}+\phi_{12}(\bar \Phi_{22}-1)e^{2i\lambda L}+\phi_{13}\bar \phi_{23}e^{2i\lambda L}]_-d\lambda,
\end{split}
\ee

\be\label{DtoNf12}
\begin{split}
f_{12}(t)=&-\frac{2}{i\pi}\int_{\partial D_3}\frac{\Sigma}{\Delta}(\lambda\phi_{13-}+i f_{02})d\lambda -\frac{1}{\pi} \int_{\partial D_3}( f_{01} \bar \Phi_{32-}+f_{02} \bar \Phi_{33-})d\lambda\\
&-\frac{4}{i\pi}\int_{\partial D_3}\frac{1}{\Delta}(\lambda \bar \Phi_{31-}-2\bar \Phi_{31}^{(1)})d\lambda +\frac{2}{\pi}\int_{\partial D_3}(f_{01}\phi_{23-}+f_{02}\phi_{33-})d\lambda\\
&-\frac{4}{i\pi}\int_{\partial D_3}\frac{\lambda}{\Delta}[(\phi_{11}-1)\bar \Phi_{31}+\phi_{12}\bar \Phi_{32}e^{2i\lambda L}+\phi_{13}(\bar \Phi_{33}-1)e^{2i\lambda L}]_-d\lambda.
\end{split}
\ee
\end{subequations}
where the conjugate of a function h denotes $\bar h=\overline {h(\bar \lambda)}$.

\item For the Neumann problem, the unknown boundary values $\{g_{01}(t), g_{02}(t)\}$ and  $\{f_{01}(t), f_{02}(t)\}$ are given by
\begin{subequations}\label{NtoDg}
\be\label{NtoDg01}
\begin{split}
g_{01}(t)=&\frac{1}{\pi}\int_{\partial D_3}\frac{\Sigma}{\Delta}\Phi_{12+}d\lambda
+\frac{2}{\pi}\int_{\partial D_3}\frac{1}{\Delta}\bar \phi_{21+}d\lambda \\
&+\frac{2}{\pi}\int_{\partial D_3}\frac{1}{\Delta}\left[\bar \phi_{21}(\Phi_{11}-1)+(\bar \phi_{22}-1)\Phi_{12}e^{-2i\lambda L}+\bar \phi_{23} \Phi_{13}e^{-2i\lambda L}\right]_+d\lambda,
\end{split}
\ee
\be\label{NtoDg02}
\begin{split}
g_{02}(t)=&\frac{1}{\pi}\int_{\partial D_3}\frac{\Sigma}{\Delta}\Phi_{13+}d\lambda
+\frac{2}{\pi}\int_{\partial D_3}\frac{1}{\Delta}\bar \phi_{31+}d\lambda \\
&+\frac{2}{\pi}\int_{\partial D_3}\frac{1}{\Delta}\left[\bar \phi_{31}(\Phi_{11}-1)+\bar \phi_{32}\Phi_{12}e^{-2i\lambda L}+(\bar \phi_{33}-1) \Phi_{13}e^{-2i\lambda L}\right]_+d\lambda,
\end{split}
\ee
\end{subequations}
and
\begin{subequations}\label{NtoDf}
\be\label{NtoDf01}
\begin{split}
f_{01}(t)=&-\frac{1}{\pi}\int_{\partial D_3}\frac{\Sigma}{\Delta}\phi_{12+}d\lambda
-\frac{2}{\pi}\int_{\partial D_3}\frac{1}{\Delta}\bar \Phi_{21+}d\lambda \\
&-\frac{2}{\pi}\int_{\partial D_3}\frac{1}{\Delta}\left[\bar \Phi_{21}(\phi_{11}-1)+(\bar \Phi_{22}-1)\phi_{12}e^{2i\lambda L}+\bar \Phi_{23} \phi_{13}e^{2i\lambda L}\right]_+d\lambda,
\end{split}
\ee

\be\label{NtoDf02}
\begin{split}
f_{02}(t)=&-\frac{1}{\pi}\int_{\partial D_3}\frac{\Sigma}{\Delta}\phi_{13+}d\lambda
-\frac{2}{\pi}\int_{\partial D_3}\frac{1}{\Delta}\bar \Phi_{31+}d\lambda \\
&-\frac{2}{\pi}\int_{\partial D_3}\frac{1}{\Delta}\left[\bar \Phi_{31}(\phi_{11}-1)+\bar \Phi_{32}\phi_{12}e^{2i\lambda L}+(\bar \Phi_{33}-1) \phi_{13}e^{2i\lambda L}\right]_+d\lambda.
\end{split}
\ee
\end{subequations}

\end{enumerate}
\end{theorem}

\begin{proof}
The representations (\ref{Sk}) follow from the relation $S(k)=e^{2i\lambda^2 T\hat{\Lambda}}\mu_2^{-1}(0,T,k)$. And the system (\ref{Phil3sys}) is the direct result of the Volteral integral equations of $\mu_2(0,t,k)$.
\begin{enumerate}
\item In order to derive (\ref{DtoNg11}) we note that equation (\ref{g112}) expresses $g_{11}$ in terms of $\Phi_{12}^{(2)}$ and $\Phi_{22}^{(1)}, \Phi_{32}^{(1)}$. Furthermore, equation (\ref{mu2x0tk}) and Cauchy theorem imply

    \[
    %\begin{split}
    -\frac{i\pi}{2}\Phi_{22}^{(1)}(t)=\int_{\partial D_2}\left(\Phi_{22}(t,\lambda )-1\right)d\lambda
    =\int_{\partial D_4}\left(\Phi_{22}(t,\lambda )-1\right)d\lambda.
    %\end{split}
    \]
    \[
    %\begin{split}
    -\frac{i\pi}{2}\Phi_{32}^{(1)}(t)=\int_{\partial D_2}\Phi_{32}(t,\lambda )d\lambda
    =\int_{\partial D_4}\Phi_{32}(t,\lambda )d\lambda.
    %\end{split}
    \]

    and
    \[
   % \begin{split}
    -\frac{i\pi}{2}\Phi_{12}^{(2)}(t)=\int_{\partial D_2}\left(\lambda \Phi_{12}(t,\lambda)-\Phi_{12}^{(1)}(t)\right)d\lambda
    =\int_{\partial D_4}\left(\lambda \Phi_{12}(t,\lambda)-\Phi_{12}^{(1)}(t)\right)d\lambda,
    %\end{split}
    \]
    Thus,
    \be\label{Phi221}
    i\pi\Phi_{22}^{(1)}=\int_{\partial D_3}\Phi_{22-}(t,\lambda)d\lambda,\qquad \quad
    %\end{equation}
%    \be\label{Phi321}
    i\pi\Phi_{32}^{(1)}=\int_{\partial D_3}\Phi_{32-}(t,\lambda)d\lambda,
    \end{equation}
    \be\label{Phi122}
    %\begin{small}
    \begin{split}
    i\pi \Phi_{12}^{(2)}(t)=&-\left(\int_{\partial D_2}+\int_{\partial D_4}\right)\left[\lambda \Phi_{12}(t,\lambda)-\Phi_{12}^{(1)}(t)\right]d\lambda\\
    =&\left(\int_{\partial D_3}+\int_{\partial D_1}\right)\left[\lambda \Phi_{12}(t,\lambda)-\Phi_{12}^{(1)}(t)\right]d\lambda\\
    =&\int_{\partial D_3}\left[\lambda \Phi_{12}(t,\lambda)-\Phi_{12}^{(1)}(t)\right]_-d\lambda \\
    =&\int_{\partial D_3^0}\left\{\lambda \Phi_{12}(t,\lambda)-\frac{ g_{01}}{2i}+\frac{2e^{-2i\lambda L}}{\Delta}[\lambda \Phi_{12}(t,\lambda)-\frac{ g_{01}}{2i}]\right\}_-d\lambda+I(t).
    \end{split}
    %\end{small}
    \ee
    where $I(t)$ is defined by
    \[
    I(t)=-\int_{\partial D_3^0}\left\{\frac{2e^{-2i\lambda L}}{\Delta}[\lambda \Phi_{12}(t,\lambda)-\frac{g_{01}}{2i}]\right\}_-d\lambda.
    \]

   The last step involves using the global relation (\ref{globalrel21}) to compute $I(t)$, that is
    \be\label{DtoNIt}
    \begin{split}
    I(t)=&\int_{\partial D_3^0}\left\{-\frac{2e^{-2i\lambda L}}{\Delta}\left[\lambda c_{12}-\Phi_{12}^{(1)}-\frac{\Phi_{12}^{(1)}\bar \phi_{22}^{(1)}+\Phi_{13}^{(1)}\bar \phi_{23}^{(1)}}{\lambda}-\bar \phi_{21}^{(1)}e^{2i\lambda L}\right]\right\}_-d\lambda\\
    &+\int_{\partial D_3^0}\left\{-\frac{2e^{-2i\lambda L}}{\Delta}\left[\frac{\Phi_{12}^{(1)}\bar \phi_{22}^{(1)}+\Phi_{13}^{(1)}\bar \phi_{23}^{(1)}}{\lambda}-(\lambda \bar \phi_{21}-\bar \phi_{21}^{(1)})e^{2i\lambda L}\right]\right\}_-d\lambda\\
    &+\int_{\partial D_3^0}\left\{\frac{2e^{-2 i\lambda L}}{\Delta}\left[ \bar \phi_{21}(\Phi_{11}-1)e^{2i\lambda L}+(\bar \phi_{22}-1) \Phi_{12}+\bar\phi_{23} \Phi_{13}\right]\right\}_-d\lambda.
    \end{split}
    \ee
    Using the asymptotic (\ref{c21largek}) and Cauchy theorem to compute the first term on the right-hand side of equation (\ref{DtoNIt}), we find
    \be\label{DtoNItres}
    \begin{split}
    I(t)=&-i\pi\Phi_{12}^{(2)}(t)+\int_{\partial D_3^0}(\frac{ g_{01}}{2i}\bar \phi_{22-}+\frac{ g_{02}}{2i}\bar \phi_{23-})d\lambda+\int_{\partial D_3^0}\frac{2}{\Delta}(\lambda \bar \phi_{21-}-2\bar \phi_{21}^{(1)})]d\lambda\\
    &+\int_{\partial D_3^0}\frac{2\lambda}{\Delta}\left[\bar\phi_{21}( \Phi_{11}-1)e^{2i\lambda L}+(\bar\phi_{22}-1) \Phi_{12}e^{-2i\lambda L}+\bar\phi_{23} \Phi_{13}e^{-2i\lambda L}\right]_-d\lambda.
    \end{split}
    \ee

    Equations (\ref{Phi122}) and (\ref{DtoNItres}) imply
    \begin{equation}\label{Phi212value}
    \begin{split}
    \Phi_{12}^{(2)}(t)=&\frac{1}{2i\pi}\int_{\partial D_3^0}\frac{\Sigma}{\Delta}(\lambda\Phi_{12-}+ig_{01})d\lambda -\frac{1}{4\pi}\int_{\partial D_3^0}(g_{01}\bar \phi_{22-}+g_{02}\bar \phi_{23-})d\lambda+\frac{1}{i\pi}\int_{\partial D_3^0}\frac{1}{\Delta}(\lambda \bar \phi_{21-}-2\bar \phi_{21}^{(1)})]d\lambda\\
    &+\frac{1}{i\pi}\int_{\partial D_3^0}\frac{2\lambda}{\Delta}\left[\bar\phi_{21}( \Phi_{11}-1)+(\bar\phi_{22}-1) \Phi_{12}e^{-2i\lambda L}+\bar\phi_{23} \Phi_{13}e^{-2i\lambda L}\right]_-d\lambda.
    \end{split}
    \end{equation}
   Equations (\ref{Phi221}) and(\ref{Phi212value}) together with (\ref{g112}) yield (\ref{DtoNg11}). Similarly, we can prove (\ref{DtoNg12}).

  \par
    The expressions (\ref{DtoNf11}) for $ f_{11}(t)$ can be derived in a similar way. Indeed, we note that equation (\ref{f112}) expresses $ f_{11}$ in terms of $\phi^{(2)}_{12}$ and $\phi_{22}^{(1)}, \phi_{32}^{(1)}$. These three equations satisfy the analog of equations (\ref{Phi221}) and (\ref{Phi122}). In particular, $\phi_{21}^{(2)}$ satisfies
    \begin{equation}\label{phi212}
    i\pi\phi_{12}^{(2)}=-\int_{\partial D_3^0}\left( \frac{\Sigma}{\Delta}(\lambda \phi_{12-}-2\phi_{12}^{(1)})\right)d\lambda +J(t),
    \end{equation}
    where
     \[
     J(t)=\int_{\partial D_3^0}\left\{\frac{2e^{2i\lambda L}}{\Delta}[\lambda \phi_{12}(t,\lambda)-\frac{ f_{01}(t)}{2i}]\right\}_-d\lambda.
    \]

   Then using the global relation  to compute $J(t)$, that is
    \be\label{DtoNJtres}
    \begin{split}
    J(t)=&-i\pi\phi_{12}^{(2)}(t)+\int_{\partial D_3^0}(\frac{ f_{01}}{2i}\bar \Phi_{22-}+\frac{ f_{02}}{2i}\bar \Phi_{23-})d\lambda-\int_{\partial D_3^0}\frac{2}{\Delta}(\lambda \bar \Phi_{21-}-2\bar \Phi_{21}^{(1)})]d\lambda\\
    &-\int_{\partial D_3^0}\frac{2\lambda}{\Delta}\left[\bar\Phi_{21}( \phi_{11}-1)+(\bar\Phi_{22}-1) \phi_{12}e^{2i\lambda L}+\bar\Phi_{23} \phi_{13}e^{2i\lambda L}\right]_-d\lambda.
    \end{split}
    \ee

    The equation (\ref{phi212}) and (\ref{DtoNJtres}) together with the asymptotics of $c_{12}(t, \lambda)$ yield (\ref{DtoNf11}). The proof of (\ref{DtoNf12}) is similar.

    \item In order to derive the representations (\ref{NtoDg01}) relevant for the Neumann problem, we note that equation (\ref{g012}) expresses $g_{01}$ and $ g_{02}$ in terms of $\Phi^{(1)}_{12}$ and $\Phi_{13}^{(1)}$, respectively. Furthermore, equation (\ref{mu2x0tk}) and Cauchy's theorem imply

        \[
        -\frac{i\pi}{2}\Phi^{(1)}_{12}(t)=\int_{\partial D_2}\Phi_{12}(t,\lambda)d\lambda=\int_{\partial D_4}\Phi_{12}(t,\lambda)d\lambda,
        \]
        Thus,
        \be\label{Phi121}
        \begin{split}
        i\pi \Phi^{(1)}_{12}(t)=&\left(\int_{\partial D_3}+\int_{\partial D_1}\right)\Phi_{12}(t,\lambda)d\lambda\\
        =&\int_{\partial D_3}\Phi_{12-}(t,\lambda)d\lambda\\
        =&\int_{\partial D_3}\left(\frac{\Sigma}{\Delta}\Phi_{12+}(t,\lambda)\right)d\lambda+K(t),
        \end{split}
        \ee
        where
        \[
        K(t)=-\int_{\partial D_3^0}\frac{2}{\Delta}\left(e^{-2i\lambda L}\Phi_{12}(t,\lambda)\right)_+ d\lambda,
        \]
        using the global relation and the asymptotic formulas of $c_{21}(t, \lambda)$, we have
        \be\label{NtoDKtres}
        %\begin{split}
        K(t)=-i\pi \Phi^{(1)}_{12}(t)+2\int_{\partial D_3^0}\left\{\frac{1}{\Delta}\bar \phi_{21+}+\frac{1}{\Delta}\left[\bar\phi_{21}( \Phi_{11}-1)e^{2i\lambda L}+(\bar\phi_{22}-1) \Phi_{12}+\bar\phi_{23} \Phi_{13}\right]_+\right\}d\lambda.
        %\end{split}
        \ee
        Equations (\ref{g012}), (\ref{Phi121}) and (\ref{NtoDKtres}) yields (\ref{NtoDg01}). The proof of the other formulas is similar.
\end{enumerate}
\end{proof}

\subsection{Effective characterizations}
Substituting into the system (\ref{Phil3sys}), (\ref{Phil1sys}) and (\ref{Phil2sys}) the expressions
\begin{subequations}
\be\label{Phij3eps}
\Phi_{ij}=\Phi_{ij,0}+\eps\Phi_{ij,1}+\eps^2\Phi_{ij,2}+\cdots,\quad i,j=1,2,3.
\ee
\be\label{phij3eps}
\phi_{ij}=\phi_{ij,0}+\eps\phi_{ij,1}+\eps^2\phi_{ij,2}+\cdots,\quad i,j=1,2,3.
\ee
\be\label{g0eps}
g_{01}=\eps g^{(1)}_{01}+\eps^2 g^{(2)}_{01}+\cdots,\\
g_{02}=\eps g^{(1)}_{02}+\eps^2 g^{(2)}_{02}+\cdots,
\ee
\be\label{f0eps}
f_{01}=\eps f^{(1)}_{01}+\eps^2 f^{(2)}_{01}+\cdots,\\
f_{02}=\eps f^{(1)}_{02}+\eps^2 f^{(2)}_{02}+\cdots,
\ee
\be\label{g1eps}
g_{11}=\eps g^{(1)}_{11}+\eps^2 g^{(2)}_{11}+\cdots,\\
g_{12}=\eps g^{(1)}_{12}+\eps^2 g^{(2)}_{12}+\cdots,
\ee
\be\label{f1eps}
f_{11}=\eps f^{(1)}_{11}+\eps^2 f^{(2)}_{11}+\cdots,\\
f_{12}=\eps f^{(1)}_{12}+\eps^2 f^{(2)}_{12}+\cdots,
\ee
\end{subequations}
where $\eps>0$ is a small parameter, we find that the terms of $O(1)$ give
\be\label{Oeps0}
O(1):\left\{
\ba{ccc}
\Phi_{13,0}=0 & \Phi_{23,0}=0 & \Phi_{33,0}=1,\\
\Phi_{11,0}=1 & \Phi_{21,0}=0 & \Phi_{31,0}=0,\\
\Phi_{12,0}=0 & \Phi_{22,0}=1 & \Phi_{32,0}=0.
\ea
\right.
\ee
Moreover, the terms of $O(\eps)$ give
\be\label{Oeps}
O(\eps):\left\{
\ba{l}
\Phi_{33,1}=0 \quad \Phi_{23,1}=0,\\[6pt]

\Phi_{13,1}(t,k)=\int_0^te^{-4i\lambda^2(t-t')}(2\lambda g^{(1)}_{02}+i g^{(1)}_{12})(t')dt',\\

\Phi_{11,1}=0,\\[6pt]

\Phi_{21,1}=\int_{0}^{t}e^{4i\lambda^2(t-t')}(i\lambda\bar g^{(1)}_{11})(t')dt',\\[6pt]

\Phi_{31,1}=\int_{0}^{t}e^{4i\lambda^2(t-t')}( i\lambda\bar g^{(1)}_{12})(t')dt',\\[6pt]

\Phi_{12,1}=\int_{0}^{t}e^{-4i\lambda^2(t-t')}(2\lambda g^{(1)}_{01}+i g^{(1)}_{11})(t')dt',\\[6pt]

\Phi_{22,1}=0,\quad \Phi_{32,1}=0.
\ea
\right.
\ee
the terms of $O(\eps^2)$ give
\be\label{Oeps2}
O(\eps^2):\left\{
\ba{l}
\Phi_{13,2}=\int_{0}^{t}e^{-4i\lambda^2(t-t')}(2\lambda g^{(2)}_{02}+i\lambda g^{(2)}_{12})(t')dt',\\[6pt]

\Phi_{23,2}=\int_{0}^{t}\left[ i\lambda \bar g^{(1)}_{11}(t')\Phi_{13,1}(t',k)-\frac{1}{2}\bar g^{(1)}_{11}g^{(1)}_{02})(t')\right]dt',\\[6pt]

\Phi_{33,2}=\int_{0}^{t}\left[i\lambda\bar g^{(1)}_{12})(t')\Phi_{13,1}(t',k)-\frac{1}{2}\bar g^{(1)}_{12}g^{(1)}_{02}(t')\right]dt',\\[6pt]

%\begin{split}
\Phi_{11,2}=\int_{0}^{t}\left[\frac{1}{2}(\bar g^{(1)}_{11} g_{01}^{(1)}+\bar g^{(1)}_{12}g_{02}^{(1)})(t')+(2\lambda g_{01}^{(1)}+i g^{(1)}_{11})\Phi_{21,1}(t',\lambda)+(2\lambda g_{02}^{(1)}+i g^{(1)}_{12})\Phi_{31,1}(t',\lambda)\right]dt',\\[6pt]

%\end{split}
\Phi_{21,2}=\int_{0}^{t}e^{4i\lambda^2(t-t')}(i\lambda \bar g^{(2)}_{11})(t')dt',\\[6pt]

\Phi_{31,2}=\int_{0}^{t}e^{4i\lambda^2 (t-t')}(i\lambda \bar g^{(2)}_{12})(t')dt',\\[6pt]

\Phi_{12,2}=\int_{0}^{t}e^{-4i\lambda^2(t-t')}(2\lambda g_{01}^{(2)}+i g^{(2)}_{11})(t')dt',\\[6pt]

\Phi_{22,2}=\int_{0}^{t}\left[(i\lambda \bar g^{(1)}_{11}(t')\Phi_{12,1}(t',k)-\frac{1}{2}\bar g^{(1)}_{11} g^{(1)}_{01})(t')\right]dt',\\[6pt]

\Phi_{32,2}=\int_{0}^{t}\left[(i\lambda \bar g^{(1)}_{12}(t')\Phi_{12,1}(t',k)-\frac{1}{2}\bar g^{(1)}_{12} g^{(1)}_{01})(t')\right]dt'.
\ea
\right.
\ee
\par
Similarly, we will have the analogue formulas for $\{\phi_{ij,l}\}_{i,j=1}^3, l=0,1,2$ expressed in terms of the boundary data at $x=L$, that is $\{f_{ij}^{(l)}\}_{i=0,1}^{j=1,2}, l=1,2$.

\par
On the other hand, expanding (\ref{DtoNg}), (\ref{DtoNf}) and assuming for simplicity that $m_{11}(\mathcal{A})(\lambda)$ has no zeros, we find
\begin{subequations}
\be\label{DtoNg11^1}
 g^{(1)}_{11}(t)=\frac{2}{i\pi}\int_{\partial D_3^0}(\lambda\Phi_{12,1-}(t,\lambda)+i g_{01}^{(1)})d\lambda+\frac{4}{i\pi}\int_{\partial D_3^0}\frac{1}{\Delta}(\lambda \bar \phi_{21,1-}(t,\lambda)-2\bar \phi_{21}^{(1)})d\lambda,
\ee
\be\label{DtoNg12^1}
 g^{(1)}_{12}(t)=\frac{2}{i\pi}\int_{\partial D_3^0}(\lambda\Phi_{13,1-}(t,\lambda)+i g_{02}^{(1)})d\lambda+\frac{4}{i\pi}\int_{\partial D_3^0}\frac{1}{\Delta}(\lambda \bar \phi_{31,1-}(t,\lambda)-2\bar \phi_{31}^{(1)})d\lambda,
\ee
\be\label{DtoNf11^1}
 f^{(1)}_{11}(t)=-\frac{2}{i\pi}\int_{\partial D_3^0}(\lambda\phi_{12,1-}(t,\lambda)+i f_{01}^{(1)})d\lambda-\frac{4}{i\pi}\int_{\partial D_3^0}\frac{1}{\Delta}(\lambda \bar \Phi_{21,1-}(t,\lambda)-2\bar \Phi_{21}^{(1)})d\lambda,
\ee
\be\label{DtoNf12^1}
 f^{(1)}_{12}(t)=-\frac{2}{i\pi}\int_{\partial D_3^0}(\lambda\phi_{13,1-}(t,\lambda)+i f_{02}^{(1)})d\lambda-\frac{4}{i\pi}\int_{\partial D_3^0}\frac{1}{\Delta}(\lambda \bar \Phi_{31,1-}(t,\lambda)-2\bar \Phi_{31}^{(1)})d\lambda,
\ee
\end{subequations}

\par
we also find that
\begin{equation}\label{Om^1}
\begin{array}{c}
\Phi_{12,1-}=4\lambda\int_0^t e^{-4i\lambda^2(t-t')} g_{01}^{(1)}(t')dt',\\
\Phi_{13,1-}=4\lambda\int_0^t e^{-4i\lambda^2(t-t')} g_{02}^{(1)}(t')dt',\\
\phi_{21,1-}=2i\lambda\int_0^t e^{4i\lambda^2(t-t')}\bar f_{11}^{(1)}(t')dt',\\
\phi_{31,1-}=2i\lambda\int_0^t e^{4i\lambda^2(t-t')}\bar f_{12}^{(1)}(t')dt'.
\end{array}
\end{equation}

\par
The Dirichlet problem can now be solved perturbatively as follows: assuming for simplicity that $m_{11}(\mathcal{A})(\lambda)$ has no zeros and given $ g_{01}^{(1)},  g_{02}^{(1)}$ and $\bar f_{11}^{(1)}, \bar f_{12}^{(1)}$, we can use equation (\ref{Om^1}) to determine $\Phi_{1j,1-}, \phi_{j1,1-}, j=2,3$. We can then compute $ g_{11}^{(1)},  g_{12}^{(1)}$ from (\ref{DtoNg11^1}), (\ref{DtoNg12^1}) and then $\Phi_{1j,1}, j=2,3$ from (\ref{Oeps}) and the analogue results for $\phi_{j1,1}, j=2,3$. In the same way we can determine $\Phi_{1j,2}, j=2,3$ from (\ref{Oeps2}) and the analogue results for $\phi_{j1,2}, j=2,3$, then compute $ g_{11}^{(2)},  g_{12}^{(2)}$ and  $ f_{11}^{(2)},  f_{12}^{(2)}$. These arguments can be extended to the higher order and also can be extended to the systems (\ref{Phil3sys}), (\ref{Phil1sys}) and (\ref{Phil2sys}) thus yields a constructive scheme for computing $S(k)$ to all orders. The construction of $S_L(\lambda)$ is similar.
\par
Similarly, these arguments also can be used to the Neumann problem. That is to say, in all cases, the system can be solved perturbatively to all orders.

\subsection{The large $L$ limit}

In the limit $L\rightarrow \infty$, the representations for $g_{11}(t),g_{12}(t)$ and $g_{01}(t),g_{02}(t)$ of theorem 4.3 \label{maintheom} reduce to the corresponding representations on the half-line. Indeed, as $L\rightarrow \infty$,
\[
\ba{l}
f_{01}\rightarrow 0,\quad f_{02}\rightarrow 0,\quad f_{11}\rightarrow 0,\quad f_{12}\rightarrow 0,\\
\phi_{ij}\rightarrow \dta_{ij},\quad \frac{\Sig}{\Dta}\rightarrow 1\mbox{ as $\lambda \rightarrow \infty$ in $D_3$}
\ea
\]
Thus, the $L\rightarrow \infty$ limits of the representations (\ref{DtoNg11}), (\ref{DtoNg12}) and (\ref{NtoDg01}), (\ref{NtoDg02}) are
\be
\ba{rl}
 g_{11}(t)&=\frac{2}{i\pi}\int_{\partial D^0_3}(\lambda \Phi_{12-}+i g_{01})d\lambda+\frac{2}{\pi}\int_{\partial D^0_3} (g_{01}\Phi_{22-}+g_{02}\Phi_{32-})d\lambda,\\%-\frac{1}{\pi}\int_{\partial D^0_3} g_{01}\bar \phi_{22-}d\lambda,\\
g_{12}(t)&=\frac{2}{i\pi}\int_{\partial D^0_3}(\lambda \Phi_{13-}+i g_{02})d\lambda+\frac{2}{\pi}\int_{\partial D^0_3} (g_{01}\Phi_{23-}+g_{02}\Phi_{33-})d\lambda.%-\frac{1}{\pi}\int_{\partial D^0_3} g_{02}\bar \phi_{33-}d\lambda.
\ea
\ee
and
\be
\ba{ll}
g_{01}(t)=\frac{1}{\pi}\int_{\partial D^0_3}\Phi_{12+}d\lambda,&
g_{02}(t)=\frac{1}{\pi}\int_{\partial D^0_3}\Phi_{13+}d\lambda,
\ea
\ee
respectively. And these formulas coincide with the corresponding half-line formulas, see  (\ref{DtoNg1}), (\ref{NtoDg0}).
\bigskip

\appendix
\section{Some formulas on the half-line}

For the convenience of reader,  we show the half-line formulas of $g_{11}(t), g_{12}(t)$ and $g_{01}(t), g_{02}(t)$ on the $\lambda-$plane.
\par
From the global relation (\ref{globalrel})and replacing $T$ by $t$, we find
\be
\mu_2(0,t,\lambda)e^{2i\lambda^2t\hat \Lam}s(\lambda)=c(t,\lambda),\quad \lambda \in(D_3\cup D_4,D_1\cup D_2,D_1\cup D_2).
\ee
We partition matrix as following,
\be\label{block}
\mu_2(0,t,\lambda)=\left(\ba{cc}\Phi_{11}&\Phi_{1j}\\\Phi_{j1}&\Phi_{2\times 2}\ea
\right),\quad j=2,3,
\ee
where $\Phi_{2\times 2}$ denotes a $2\times 2$ matrix, $\Phi_{1j}$ denotes a $1\times 2$ vector, $\Phi_{j1}$ denotes a $2\times 1$ vector.
Then, we can write the second column of the global relation, undering the matrix partitioned as (\ref{block}), as
\begin{subequations}\label{globalrelsec}
%\be\label{globalrel11}
%\Phi_{11}(t,k)s_{11}(k)+\Phi_{1j}(t,k)s_{j1}(k)e^{4ik^2t}=c_{11}(t,k),\quad k\in D_3\cup D_4,
%\ee
\be\label{globalrel1j}
\Phi_{11}(t,\lambda)s_{1j}(\lambda)s^{-1}_{2\times 2}(\lambda)e^{-4i\lambda^2t}+\Phi_{1j}(t,\lambda)=c_{1j}(t,\lambda),\quad \lambda\in D_1\cup D_2,
\ee
%\be\label{globalrelj1}
%\Phi_{j1}(t,k)s_{11}(k)+\Phi_{2\times 2}(t,k)s_{j1}(k)e^{4ik^2t}=c_{j1}(t,k),\quad k\in D_3\cup D_4,
%\ee
\be\label{globalrel2by2}
\Phi_{j1}(t,\lambda)s_{1j}(\lambda)s^{-1}_{2\times 2}(\lambda)e^{-4i\lambda^2t}+\Phi_{2\times 2}(t,\lambda)=c_{2\times 2}(t,\lambda),\quad \lambda\in D_1\cup D_2,
\ee
\end{subequations}
The functions $c_{1j}(t,\lambda),c_{2\times 2}(t,\lambda)$ are analytic and bounded in $D_1\cup D_2$ away from the possible zeros of $m_{11}(\lambda)$ and of order $O(\frac{1}{\lambda})$ as $k\rightarrow \infty$.
\par
From the asymptotic of $\mu_j(x,t,\lambda)$ in (\ref{mujasykinf}) we have
\be\label{mu2x0tk}
\ba{rl}
\mu_{2}(0,t,\lambda)=&\id+\frac{1}{\lambda}\left(\ba{cc}\int_{(0,0)}^{(0,t)}\Dta_{11}dx'+\eta_{11}dt'&\frac{1}{2i}Q\\\frac{1}{8i}|q|^2 Q^{T}-\frac{1}{4}\bar Q^T_x&\int_{(0,0)}^{(0,t)}\Dta dx'+\eta dt'\ea\right)\\
&{}+\frac{1}{\lambda^2}\left(\ba{cc}\mu_{11}^{(2)}&\frac{1}{4}Q_x+\frac{1}{2i}Q\mu_{2\times2}^{(1)}\\\mu_{j1}^{(2)}&\mu_{2\times2}^{(2)}\ea\right)+O(\frac{1}{\lambda^3})
\ea
\ee
where $Q=(q_1,q_2)$, $\Dta_{11}$ is defined by first identities of (\ref{Dtadef}), $\eta_{11}$ is defined by (\ref{etadef}), $\Dta$ and $\eta$ are $2\times 2$ matrices defined as following,
\be
\Dta=\left(\ba{cc}\Dta_{22}&\Dta_{23}\\\Dta_{32}&\Dta_{33}\ea\right),\qquad  \eta=\left(\ba{cc}\eta_{22}&\eta_{23}\\\eta_{32}&\eta_{33}\ea\right),
\ee
%here, $\{\Dta^{(j)}_{kl}\}_{k,l=2,3},j=1,2$ are defined as (\ref{Dtadef}) and (\ref{Dta2def}).
Also, we have
\begin{subequations}\label{Phi3}
\be\label{Phi13}
\Phi_{1j}(t,\lambda)=\frac{\Phi_{1j}^{(1)}(t)}{\lambda}+\frac{\Phi_{1j}^{(2)}(t)}{\lambda^2}+O(\frac{1}{\lambda^3}),\quad \lambda\rightarrow \infty,\lambda\in D_1\cup D_2
\ee
\be\label{Phi2by2}
\Phi_{2\times 2}(t,\lambda)=\id_{2\times 2}+\frac{\Phi_{2\times 2}^{(1)}(t)}{\lambda}+\frac{\Phi_{2\times 2}^{(2)}(t)}{\lambda^2}+O(\frac{1}{\lambda^3}),\quad \lambda\rightarrow \infty,\lambda\in D_1\cup D_2.
\ee
where
\[
\ba{ll}
\Phi_{1j}^{(1)}(t)=\frac{1}{2i}g_0(t),&\Phi_{1j}^{(2)}(t)=\frac{1}{4}g_1(t)-\frac{i}{2}g_0\Phi_{2\times 2}^{(1)}(t)\\
\Phi_{2\times 2}^{(1)}(t)=\int_0^t\eta dt'.
\ea
\]
here $g_0(t)$ and $g_1(t)$ are vector boundary functions defined by the boundary data of (\ref{ibv-cgi}) as $g_0(t)=(g_{01}(t),g_{02}(t))$ and $g_1(t)=(g_{11}(t),g_{12}(t))$.
\end{subequations}
\par
In particular, we find the following expressions for the boudary values:
\begin{subequations}\label{g01}
\be\label{g0}
g_0=2i\Phi_{1j}^{(1)}(t),
\ee
\be\label{g1}
g_1=2ig_0\Phi_{2\times 2}^{(1)}(t)+4\Phi_{1j}^{(2)}(t),
\ee
\end{subequations}
We will also need the asymptotic of $c_{1j}(t,\lambda)$,
\begin{lemma}
The global relation (\ref{globalrelsec}) implies that the large $\lambda$ behavior of $c_{1j}(t,\lambda),c_{2\times 2}(t,\lambda)$ satisfies
\be\label{cjlargek}
c_{1j}(t,\lambda)=\frac{\Phi_{1j}^{(1)}(t)}{\lambda}+\frac{\Phi_{1j}^{(2)}(t)}{\lambda^2}+O(\frac{1}{\lambda^3}),\quad \lambda\rightarrow \infty,\lambda\in D_1.
\ee

\end{lemma}
\begin{proof}
Analogous to the proof provided in Lemma 4.2.
\end{proof}
We can now derive the maps between Dirichlet boundary condition and  the Neumann boundary condition as follows:
\begin{enumerate}
\item For the Dirichlet problem, the unknown Neumann boundary value $g_1(t)$ is given by

\be\label{DtoNg1}
\ba{rcl}
g_1(t)&=&\frac{2}{\pi i}\int_{\partial D_3}(\lambda\Phi_{1j-}(t,\lambda)+ig_0(t))+\frac{2g_0}{\pi}\int_{\partial D_3}\Phi_{2\times2-}d\lambda\\
&&{}-\frac{4}{i\pi}\int_{\partial D_3}\lambda e^{-4i\lambda^2t}\Phi_{11}(-\lambda)s_{1j}(-\lambda)s^{-1}_{2\times 2}(-\lambda)d\lambda.
\ea
\ee

\item For the Neumann problem, the unknown boundary values $g_0(t)$ is given by

\be\label{NtoDg0}
\ba{rl}
g_0(t)=&\frac{1}{\pi}\int_{\partial D_3}\Phi_{1j+}(t,\lambda)d\lambda+\frac{2}{\pi}\int_{\partial D_3}e^{-4i\lambda^2t}\Phi_{11}(-\lambda)s_{1j}(-\lambda)s^{-1}_{2\times 2}(-\lambda)d\lambda.
\ea
\ee

\end{enumerate}
%\end{theorem}
\begin{proof}
%The representations (\ref{Sk}) follow from the relation $S(k)=e^{4ik^2T}\mu_2^A(0,T,k)^T$. And the system (\ref{Phil3sys}) is the direct result of the Volteral integral equations of $\mu_2(0,t,k)$.
\begin{enumerate}
\item In order to derive (\ref{DtoNg1}) we note that equation (\ref{g1}) expresses $g_1$ in terms of $\Phi_{2\times 2}^{(1)}$ and $\Phi_{1j}^{(2)}$. Furthermore, equation (\ref{Phi3}) and Cauchy theorem imply
    \[
    -\frac{\pi i}{2}\Phi_{2\times 2}^{(1)}(t)=\int_{\partial D_2}[\Phi_{2\times 2}(t,\lambda)-\id_{2\times 2}]d\lambda=\int_{\partial D_4}[\Phi_{2\times 2}(t,\lambda)-\id_{2\times 2}]d\lambda
    \]
    and
    \[
    -\frac{\pi i}{2}\Phi_{1j}^{(2)}(t)=\int_{\partial D_2}\left[\lambda\Phi_{1j}(t,\lambda)-\frac{g_0(t)}{2i}\right]d\lambda=\int_{\partial D_4}\left[\lambda\Phi_{1j}(t,\lambda)-\frac{g_0(t)}{2i}\right]d\lambda.
    \]
    Thus,
    \be\label{Phi331}
    \ba{l}
    i\pi\Phi_{2\times 2}^{(1)}(t)=-\left(\int_{\partial D_2}+\int_{\partial D_4}\right)[\Phi_{2\times 2}(t,\lambda)-\id_{2\times 2}]d\lambda\\
    {}=\left(\int_{\partial D_1}+\int_{\partial D_3}\right)[\Phi_{2\times 2}(t,\lambda)-\id_{2\times 2}]d\lambda\\
    {}=\int_{\partial D_3}[\Phi_{2\times 2}(t,\lambda)-\id_{2\times 2}]d\lambda-\int_{\partial D_3}[\Phi_{2\times 2}(t,-\lambda)-\id_{2\times 2}]d\lambda\\
    {}=\int_{\partial D_3}\Phi_{2\times 2-}(t,\lambda)d\lambda.
    \ea
    \ee
    Similarly,
    \be\label{Phi132}
    \ba{l}
    i\pi \Phi_{1j}^{(2)}(t)=\left(\int_{\partial D_3}+\int_{\partial D_1}\right)\left[\lambda\Phi_{1j}(t,\lambda)-\frac{g_0(t)}{2i}\right]d\lambda\\
    =\left(\int_{\partial D_3}-\int_{\partial D_1}\right)\left[\lambda\Phi_{1j}(t,-\lambda)-\frac{g_0(t)}{2i}\right]d\lambda+I(t)\\
    =\int_{\partial D_3}\left[\lambda\Phi_{1j-}(t,\lambda)+ig_0(t)\right]d\lambda+I(t).
    \ea
    \ee
    where $I(t)$ is defined by
    \[
    I(t)=2\int_{\partial D_1}\left[\lambda\Phi_{1j}(t,\lambda)-\frac{g_0(t)}{2i}\right]d\lambda
    \]
    The last step involves using the global relation to compute $I(t)$
    \be\label{DtoNIt}
    \ba{l}
    I(t)=2\int_{\partial D_1}\left[\lambda(c_{1j}s^{-1}_{2\times 2}-\Phi_{11}s_{1j}s^{-1}_{2\times 2}e^{-4i\lambda^2t})-\frac{g_0(t)}{2i}\right]d\lambda\\
    \ea
    \ee
    Using the asymptotic (\ref{cjlargek}) and Cauchy theorem to compute the first term on the right-hand side of equation (\ref{DtoNIt}), we find
    \be\label{DtoNItres}
    \ba{l}
    I(t)=-i\pi \Phi^{(2)}_{1j}-2\int_{\partial D_3}\lambda\Phi_{11}(-\lambda)s_{1j}(-\lambda)s^{-1}_{2\times 2}(-\lambda)e^{-4i\lambda^2t}d\lambda.
    \ea
    \ee
    Equations (\ref{Phi132}) and (\ref{DtoNItres}) imply
    \[
    \ba{l}
    \Phi_{1j}^{(2)}(t)=\frac{1}{2\pi i}\int_{\partial D_3}\left[\lambda\Phi_{1j-}(t,\lambda)+ig_0(t)\right]d\lambda\\
    {}-\frac{1}{\pi i}\int_{\partial D_3}\lambda\Phi_{11}(-\lambda)s_{1j}(-\lambda)s^{-1}_{2\times 2}(-\lambda)e^{-4i\lambda^2t}d\lambda.
    \ea
    \]
    This equation together with (\ref{g1}) and (\ref{Phi331}) yields (\ref{DtoNg1}).
    \par

    \item In order to derive the representations (\ref{NtoDg0}) relevant for the Neumann problem, we note that equation (\ref{g0}) expresses $g_0$ in terms of $\Phi^{(1)}_{1j}$. Furthermore, equation (\ref{Phi13}) and Cauchy's theorem imply

        \be\label{Phi131}
        \ba{l}
        -\frac{\pi i}{2}\Phi^{(1)}_{1j}(t)=\int_{\partial D_2}\Phi_{1j}(t,\lambda)d\lambda=\int_{\partial D_4}\Phi_{1j}(t,\lambda)d\lambda,
        \ea
        \ee
        Thus,
        \be\label{Phi1j1}
        \ba{l}
        i\pi \Phi^{(1)}_{1j}(t)=\left(\int_{\partial D_3}+\int_{\partial D_1}\right)\Phi_{1j}(t,\lambda)d\lambda\\
        {}=\left(\int_{\partial D_3}-\int_{\partial D_1}\right)\Phi_{1j}(t,\lambda)d\lambda+2\int_{\partial D_1}\Phi_{1j}(t,\lambda)d\lambda\\
        {}=\int_{\partial D_3}\Phi_{1j+}(t,\lambda)d\lambda+2\int_{\partial D_1}\Phi_{1j}(t,\lambda)d\lambda,
        \ea
        \ee
        and using the global relation, we have
        \be\label{Phi1jtk}
        \ba{l}
        2\int_{\partial D_1}\Phi_{1j}(t,\lambda)d\lambda=2\int_{\partial D_1}(c_{1j}s^{-1}_{2\times 2}-\Phi_{11}s_{1j}s^{-1}_{2\times 2}e^{-4i\lambda^2t})d\lambda\\
        {}=-i\pi \Phi^{(1)}_{1j}(t)+2\int_{\partial D_3}\Phi_{11}(-\lambda)s_{1j}(-\lambda)s^{-1}_{2\times 2}(-\lambda)e^{-4i\lambda^2t}d\lambda.
        \ea
        \ee
        Equations (\ref{g0}), (\ref{Phi1j1}) and (\ref{Phi1jtk}) yields (\ref{NtoDg0}).
\end{enumerate}
\end{proof}

\bigskip

{\bf Acknowledgements.}
Fan was support by grants from the National Science Foundation of China under Project No. 11671095.   Xu was supported by National Science Foundation of China under project No.11501365,
Shanghai Sailing Program supported by Science and Technology Commission of Shanghai Municipality under Grant No.15YF1408100 and the Hujiang Foundation of China (B14005).

\end{document}